\documentclass[onecolumn,draftcls]{IEEEtran}

\usepackage{multirow}
\usepackage{tikz,ifthen,calc,graphicx,pgfplots}
\usepackage{amsmath}
\usepackage{amsfonts,amssymb}
\usepackage{url,breakurl}
\usepackage{cite}
\usepackage{mathrsfs}
\usepackage{color,bm}
\usepackage[ruled,vlined]{algorithm2e}

\newtheorem{definition}{Definition}[section]
\newtheorem{theorem}{Theorem}
\newtheorem{lemma}{Lemma}

\usepackage{caption}
\usepackage{subcaption}

\def\D{{\mathbf D}}

\def\R{{\mathbf R}}

\def\J{{\mathbf J}}

\def\U{{\mathbf U}}

\def\S{{\mathbf S}}

\def\x{{\mathbf x}}
\def\bk{{\mathbf k}}
\def\y{{\mathbf y}}

\def\k{{\mathbf k}}

\def\h{{\mathbf h}}
\def\f{{\mathbf f}}

\def\g{{\mathbf g}}

\def\n{{\mathbf n}}

\def\q{{\mathbf q}}

\def\0{{\mathbf 0}}

\def\Sigmab{{\boldsymbol \Sigma}}
\def\lambdab{{\boldsymbol \lambda}}

\def\xib{{\boldsymbol \xi}}

\newcommand{\cov}{\ensuremath{\mathrm{cov}}}
\newcommand{\bn}{\ensuremath{\mathbf{n}}}

\title{Detecting Directionality in Random Fields Using the Monogenic Signal}
\author{Sofia Olhede, David Ram\'irez,~\IEEEmembership{Member,~IEEE} and Peter J. Schreier,~\IEEEmembership{Senior Member,~IEEE}
\thanks{This work was presented in part at the 2012 IEEE International Conference on Image Processing (ICIP).}
}

\begin{document}

\maketitle

\begin{abstract}
Detecting and analyzing directional structures in images is important in many applications since one-dimensional patterns often correspond to important features such as object contours or trajectories. Classifying a structure as directional or non-directional requires a measure to quantify the degree of directionality and a threshold, which needs to be chosen based on the statistics of the image. In order to do this, we model the image as a random field. So far, little research has been performed on analyzing directionality in random fields. In this paper, we propose a measure to quantify the degree of directionality based on the random monogenic signal, which enables a unique decomposition of a 2D signal into local amplitude, local orientation, and local phase. We investigate the second-order statistical properties of the monogenic signal for isotropic, anisotropic, and unidirectional random fields. We analyze our measure of directionality for finite-size sample images, and determine a threshold to distinguish between unidirectional and non-unidirectional random fields, which allows the automatic classification of images.
\end{abstract}
\vspace{0.1in}

\section{Introduction}

The detection and analysis of directional structure in images is crucial to many applications since one-dimensional patterns often correspond to important image features such as object contours or trajectories. Detecting one-dimensional patterns and estimating their orientation is particularly important; see, for instance, the detection of ship wakes in Synthetic Aperture Radar (SAR) images \cite{Copeland:1995}, optical flow estimation \cite{Felsberg:05,Biguen:1991} and its application to myocardial motion estimation \cite{Alessandrini:2013}, the analysis of texture by estimating multidimensional orientation \cite{Biguen:1991}, the efficient coding of local differential structures in images \cite{Silvan:1999}, image encoding, labeling and reconstruction \cite{Barth:1993}, and the analysis of superimposed directional patterns, which may occur in X-ray projection imaging \cite{Aach:2006}.

Many of the techniques developed for orientation estimation assume that there is indeed a directional structure to be estimated. In the absence of such a directional structure, they will therefore still produce an estimated direction---which would be meaningless. In order to address this problem, several measures for the degree of directionality have been defined: In \cite{Krieger1996} the authors propose a measure called ``intrinsic dimensionality,'' which is related to the degrees of freedom of an image \cite{zetzsche:1990}. The intrinsic dimensionality of a constant image is zero; if the image can be expressed as a univariate function of a linear combination of the coordinates, its dimensionality is one; otherwise it is two. Another measure related to the intrinsic dimensionality is the Gaussian curvature \cite{spivak1975comprehensive}. A different measure aimed at detecting edges and corners is presented in \cite{Harris88acombined}. Finally, \cite{multiresolution_monogenic} introduces a measure called ``coherency'' because it is normalized to values between 0 and 1. The coherency is closely related to the measure proposed in our paper, yet \cite{multiresolution_monogenic} does not provide a detailed statistical analysis. Such an analysis is needed because having a normalized measure for the degree of directionality alone may not help us classify a structure as directional or nondirectional: For example, is a structure with degree of directionality equal to, say, 0.6 directional or not? A meaningful answer to this question can only be given based on sound statistical arguments, for a specified null model. Moreover, having a threshold based on statistical arguments will allow us to \emph{automatically} classify images, which is essential when dealing with large amounts of data, for instance, in the detection of ship wakes in the ocean \cite{Copeland:1995}. In order to derive a threshold, we model the image as a {\em stationary random field} and then phrase our problem as a hypothesis test: ``Is there a directional structure or not?'' Deciding this question requires not only a measure for directionality but also a {\em threshold}, based on the statistics of the image, above which a structure can indeed be regarded as directional. 

There has been little work on detecting and analyzing directional structures in random fields. A random field is a stochastic process whose argument is a multidimensional vector. In our case, the argument is a 2D vector, and a spatial 2D random field may also be called a random image. Random fields are useful models for applications in areas as diverse as geophysics \cite{Goff,Schoenberg}, oceanography \cite{Sokolov:1999}, and medical imaging \cite{Karssemeijer}. A fundamental characteristic of a spatial random field is the degree of rotational invariance of its second-order statistical properties \cite[p.~57]{Christakos1992}. A spatially {\em isotropic} random field exhibits perfect rotational invariance as its spatial covariance displays circular contour levels. Fields that are not isotropic are called {\em anisotropic}. A common subclass is the class of {\em geometrically anisotropic} fields \cite[p.~61]{Christakos1992} that have covariances with elliptical rather than circular contours. Taking this idea to the extreme, we arrive at {\em unidirectional} random fields, where there exists a rotation so that there is variation only in one of the two axes. Figure \ref{fig:random_fields} shows samples of three random images: isotropic, geometrically anisotropic, and unidirectional. This shows the kinds of features present in a field as it becomes more anisotropic.

A large fraction of the work in statistics has focussed on isotropic random fields---which is at least partially due to the fact that these have convenient mathematical properties---and many spatial models utilize isotropy \cite{Cressie:1991}. Yet isotropic random fields are obviously unsuitable to model directional structures. 
In order to deal with anisotropic random fields, there are a number of models \cite{Cressie:1991,Christakos1992,Bonami:2003}, the most common of which is geometric anisotropy. However, for highly directional random fields, it would seem that the most appropriate model should be one that consists of unidirectional components. 

In the deterministic case, there exist quite a few papers that deal with the estimation of local orientation, e.g., \cite{Aach:2006,Felsberg:05,Biguen:1991,Copeland:1995,DiClaudio:2010,Silvan:1999,multiresolution_monogenic}. For a brief summary of some elementary techniques, we refer the reader to Section \ref{sec:previous}. We use an approach based on the {\em random monogenic signal}, which allows us to define a statistical measure for the degree of unidirectionality and to construct statistical tests for the presence of directional structure. The monogenic signal \cite{Felsberg2001} enables the unique decomposition of a two-dimensional real image $f(\x)$, with $\x=\left[x_1,x_2\right]^T$, into a {\em local} amplitude, a {\em local} orientation, and a {\em local} phase. It is arguably the most compelling 2D-generalization of the analytic signal, which enables the unique decomposition of a real signal $x(t) = a(t)\cos\phi(t)$ into local amplitude $a(t)$ and local phase $\phi(t)$. The monogenic signal has received considerable attention, with applications in image processing ranging from contour detection and local structure analysis \cite{Felsberg:2004,Zang:2007}, stereo, motion estimation, and image registration \cite{Mellor:2005,felsberg2002disparity} to image segmentation and phase-contrast imaging \cite{Ali:2008}. There have also been extensions of the concept to a multiresolution monogenic signal in the wavelet domain \cite{Olhede2009,multiresolution_monogenic}.

As we will review in Section \ref{sec:prelim}, the monogenic signal is constructed by complementing the original signal $f(\x)$ with its two Riesz transforms $g(\x)$ and $h(\x)$ \cite{Stein}. It can either be represented as a three-dimensional vector $[f(\x), g(\x), h(\x)]^{T}$ or as a quaternion $m(\x) = f(\x) + ig(\x) + jh(\x) + k\cdot 0$, where the $k$-part remains zero. So far, work on monogenic signals has focussed only on the deterministic case. Aside from our own conference paper \cite{Olhede:2012}, we are not aware of any work that has been performed on a {\em random} monogenic signal. In Section \ref{sec:sec-ord-char}, we investigate the second-order statistical properties of random monogenic signals for stationary random fields, and in Section \ref{sec:iso-direct}, we examine these properties for the special and important cases of isotropic, geometrically anisotropic, and unidirectional random fields.

In Section \ref{sec:test}, we introduce a measure to quantify the degree of unidirectionality for a random field. A related local measure of unidirectionality, which is appropriate for deterministic signals, has been defined in  \cite{multiresolution_monogenic}. We also provide a thorough statistical analysis of our measure. In particular, we show that, for an infinite-size sample image, it is identically one only for unidirectional random fields and zero for isotropic random fields, with values in between for other degrees of anisotropy. For a finite-size sample, this measure of directionality is no longer guaranteed to be one. We determine its finite sample expectation, by carefully expanding the properties depending on the size of the image, and a threshold to distinguish between unidirectional and non-unidirectional fields. Finally, in Section \ref{sec_numres}, we illustrate the performance of our measure on simulated and real random fields. 

\section{Preliminaries}
\label{sec:prelim}

\subsection{Quaternion random vectors}

In the one-dimensional case, it is common practice to encode the signal and its Hilbert transform in one complex-valued analytic signal. While this is not strictly necessary (obviously one could also work with a 2D vector instead), this practice is universally accepted because it illuminates and simplifies matter significantly. If we would like to do something similar with the monogenic signal, we need to employ the algebra of quaternions. 

In this section, we provide a brief review of quaternion algebra and the second-order analysis of quaternion random vectors (see, e.g., \cite{Via2010}). Quaternions are 4D hypercomplex numbers, first proposed by Hamilton \cite{hamilton_quaternions}, and defined as
\begin{equation}
q = a + b i + c j+ d k, 
\end{equation}
where $a$, $b$, $c$, and $d$ are real numbers and $i$, $j$, and $k$ are imaginary units satisfying
\begin{align}
i j &= k = - j i &
j k &= i = -k j \\
k i &= j = -i k &
i^2 &= j^2 = k^2 = i j k = -1.
\end{align}
It is easy to check that quaternions form an algebra $\mathbb{H}$ that is non-commutative, i.e., for $q_1, q_2 \in \mathbb{H}$, generally $q_1 q_2 \neq q_2 q_1$. The conjugate of $q$ is defined as $q^{\ast}= a - b i - c j -d k$ and the norm is $|q| = \sqrt{q q^{\ast}} = \sqrt{a^2 + b^2 + c^2 + d^2}$, which satisfies $|q_1 q_2| = |q_1| |q_2|$. The inverse of $q$, for $q \neq 0$, is $q^{-1} = q^{\ast}/|q|^2$, the inner product between $q_1$ and $q_2$ is $\text{Re}(q_1 q_2)$, and two quaternions are orthogonal if their inner product is zero. Finally, the involution of $q$ over a pure unit quaternion $\eta$ is given by $q^{(\eta)} = -\eta q \eta$. For a more complete review of quaternions, we refer the reader to \cite{ward_quaternions}.
        
\label{sec:rev_quat}

The second-order statistical analysis of a zero-mean {\em complex} random vector $\x$ is based on the covariance matrix $\R_{\x, \x} = \cov\{\x, \x\} = E[\x \x^H]$ and the complementary covariance matrix $\R_{\x, \x^{\ast}} = \cov\{\x, \x^{\ast}\} = E[\x \x^T]$ \cite{book_peter}. A common assumption in complex-valued signal processing is propriety, which is characterized by vanishing complementary covariance. This can be visualized as rotational invariance because, in the proper case, $\x$ and $\x e^{j\alpha}$ have the same second-order moments for arbitrary real angle $\alpha$. While propriety can often be justified, there are also many situations where it is a very poor model of the underlying physics \cite{book_peter}.

To completely characterize the second-order statistics of a zero-mean {\em quaternion} random vector $\q$, we need the covariance matrix and \emph{three} complementary covariance matrices \cite{Via2010}. There is some freedom in how to choose these complementary covariance matrices. We will employ the most useful choice
\begin{align}
\R_{\q, \q} &= \cov\{\q, \q\}, & \R_{\q, \q^{(\eta)}}, &= \cov\{\q, \q^{(\eta)}\} & \R_{\q, \q^{(\eta')}}, &= \cov\{\q, \q^{(\eta')}\} & \R_{\q, \q^{(\eta'')}} &= \cov\{\q, \q^{(\eta'')}\},
\end{align}
where $\eta, \eta'$ and $\eta''$ are three {\em orthogonal} pure unit quaternions, for instance, $\eta = i$, $\eta' = j$, $\eta'' = k$. Because there are three complementary covariance matrices, there are different kinds of propriety for quaternion random vectors \cite{Via2010}. The only kind of interest to us is $\mathbb{C}^{\eta}$-propriety. A quaternion random vector is $\mathbb{C}^{\eta}$-proper if and only if both $\R_{\q, \q^{(\eta')}}$ and $\R_{\q, \q^{(\eta'')}}$ vanish. This obviously depends on an appropriate choice of $\eta$. If $\q$ is $\mathbb{C}^{\eta}$-proper for some $\eta$, it is generally improper for a different choice of $\eta$. In our case, the choice of $\eta$ will be linked to the directional structure of the random field.

\subsection{Monogenic signal}

The monogenic signal was introduced by Felsberg and Sommer in \cite{Felsberg2001}, but it had already seen some prior use in applied mathematics \cite{Dixon,Duffin} and geophysics \cite{Nabighian}. It is a convenient method of defining an amplitude and a vector-valued phase at any point in space \cite{Felsberg2001}, and can be considered an appropriate generalization of the analytic signal \cite{Stein}. The continuous monogenic signal is defined using the continuous Riesz transform \cite{Stein}. The Riesz transform enjoys a number of convenient mathematical properties, chief among them the commutativity with spatial translations and dilations, and equivariance with respect to rotation \cite{Stein}. These properties make the Riesz transform a very compelling candidate for separating information about structural and energetic aspects of a 2D signal. There are two Riesz transforms in 2D, which are defined by\footnote{By default, when we do not state the limits of an integral, they are $-\infty$ and $\infty$.}
\begin{equation}
\mathcal{R}^{(l)} f(\x) = \frac{1}{2 \pi} \iint f(\y) \frac{y_l-x_l}{\| \x - \y\|^3}\,d \y=(r^{(l)} \ast f)(\x), \quad l = 1,2,
\label{Riesz_eq}
\end{equation} 
where $r^{(l)}(\x) = - x_l/(2 \pi \|\x\|^3)$, $\x = [x_1, x_2]^T$, and $\y = [y_1, y_2]^T$. The monogenic signal is formed by placing the two Riesz transforms into the $i$- and $j$-parts of a quaternion-valued signal, leaving the $k$-part empty:
\begin{equation}
m(\x) = \mathcal{M}f(\x)=f(\x)+i g(\x) + j h(\x),
\end{equation}
where $g(\x) = \mathcal{R}^{(1)}f(\x)$ and $h(\x) = \mathcal{R}^{(2)}f(\x)$. The Riesz transforms are most easily described in the 2D Fourier domain, where the Riesz transform kernels are defined by
\begin{align}
\label{Riesz2}
R^{(1)}(\bk)&=  -i \frac{k_1}{\|\bk\|} = -i\cos(\kappa), \\
R^{(2)}(\bk)&=  -i \frac{k_2}{\|\bk\|} = -i\sin(\kappa). \label{Riesz2_2}
\end{align} 
In this equation, $\k = [k_1, k_2]^T = k [\cos(\kappa), \sin(\kappa)]^T$ is the 2D wavenumber, and $\kappa = \arg \k$ is the corresponding angle in the 2D plane. If the transformation is implemented over the entire real plane, the monogenic signal can be calculated either in the spatial domain using \eqref{Riesz_eq} or, equivalently, in the 2D Fourier domain using \eqref{Riesz2} and \eqref{Riesz2_2}.

For a sampled lowpass image $\{f_{n,n'}=f(\x_{n,n'}), \, n, n' = -N/2, \ldots, N/2-1\}$,\footnote{Without loss of generality, we assume that $N$ is even.} assumed w.l.o.g. to be sampled at a unit sampling period, the Riesz transforms have to be calculated discretely. Analogously to the discrete-time Fourier transform (DTFT) and the discrete Fourier transform (DFT), we need to define both a discrete-space Riesz transform, which is continuous in the wavenumber domain, and a discrete Riesz transform, which is discrete and periodic in both space and wavenumber domains. To emphasize the difference between the two, we sometimes refer to the discrete Riesz transform as the periodic discrete Riesz transform. 

The impulse response of the {\em discrete-space Riesz transform} for a non-periodic discrete-space signal is
\begin{equation}
{r}^{ (l)} (\x_{n, n'}) = - \frac{i}{(2 \pi)^2} \int_{-\pi}^{\pi} \int_{-\pi}^{\pi} \frac{k_l}{\sqrt{k_1^2+k_2^2}}
e^{i (k_1n+k_2 n')}\,d\bk, \quad l = 1,2.
\end{equation}
The first Riesz transform of the random field $f_{n,n'}$ is obtained as the convolution
\begin{equation} 
g (\x_{n, n'}) = \sum_{l,l'=-\infty}^{\infty} {r}^{ (1)} (\x_{l-n,l'-n'}) f (\x_{l, l'}).
\end{equation}
and the second Riesz transform $h (\x_{n, n'})$ is obtained analogously by convolving $f (\x_{n, n'})$ with ${r}^{(2)} (\x_{n,n'})$. Using the 2D DFT, given by
\begin{equation}
F(\bk)= \sum_{n, n' = -N/2}^{N/2-1} f (\x_{n, n'}) e^{- i (k_1n+k_2n')},
\end{equation}
with $k_l$ uniformly spaced in $\left(- \pi, \pi \right]$, the periodic impulse response of the {\em discrete Riesz transform} is defined by
\begin{equation}
\tilde{r}^{(l)} (\x_{n, n'}) = - \frac{ i}{ N^2} \sum_{\bk} \frac{k_l}{\sqrt{k_1^2+k_2^2}}
e^{i (k_1n+k_2 n')}, \quad l = 1,2.
\end{equation}
The first Riesz transform of $f(\x_{n,n'})$ is then obtained as the {\em circular} convolution  
\begin{equation} 
\tilde{g} (\x_{n, n'}) =\sum_{l, l' = -N/2}^{N/2-1} \tilde{r}^{(1)} (\x_{l-n \bmod{N},l'-n' \bmod{N}}) f(\x_{l, l'}),
\end{equation}
where $l-n \bmod N$ denotes $l-n$ modulo $N$ and the second Riesz transform $\tilde{h}(\x_{n,n'})$ is obtained analogously by circularly convolving $f (\x_{n, n'})$ with $\tilde{r}^{(2)} (\x_{n,n'})$. In practice, if implemented numerically, we always calculate the periodic discrete Riesz transform. This has a direct implication on the measure of unidirectionality we propose, because it introduces an additional error due to periodic filtering.

\section{Second-order statistical characterization of the monogenic signal}
\label{sec:sec-ord-char}

Consider a zero-mean random field $f(\x)$ with covariance function $r_{ff}(\x,\xib) = \cov\{f(\x),f(\x-\xib)\}$, where $\cov(a,b) = E[ab^{\ast}]$ is the covariance operator, $\x$ is a global spatial position and $\xib$ a local spatial offset. If $r_{ff}(\x,\xib) = r_{ff}(\xib), \; \forall \, \x$, i.e., the covariance function only depends on the spatial offset $\xib$, then $f(\x)$ is called (wide-sense) stationary. Nonstationary random fields are difficult to analyze, and most of the properties we will derive do not hold for this general case. Moreover, many nonstationary random fields of interest can be locally approximated as stationary for a small enough patch \cite{priestley1988non}. Therefore, from now on we only consider stationary random fields, which is also the kind of random field that most of the statistical literature focuses on \cite{Stein1999}. 

In order to characterize the second-order moments of the random monogenic signal, we obtain the covariances of the random field and its Riesz transforms. To do so, let us consider the spectral representation \cite{loeve} of the random field $f(\x)$, whose spectral process is ${\cal Z}_f(\bk)$. The spectral process is a complex-valued random measure with uncorrelated increments (since $f(\x)$ is stationary), i.e.,  $\cov\{d{\cal Z}_f(\bk),d{\cal Z}_f(\bk')\}=S_{ff}(\bk)\delta(\bk-\bk')\,d\bk \, d\bk'$, where $S_{ff}(\bk)$ is the power spectral density of $f(\x)$. Then, the spectral representation of $f(\x)$ is
\begin{equation}
f(\x)= \iint d{\cal Z}_f(\bk) e^{i \bk^T\x}.
\end{equation}
The covariance function and the power spectral density are therefore 2D Fourier transform pairs:
\begin{equation}
r_{ff}(\xib) =   \iint S_{ff}(\bk) e^{i \bk^T\xib} \, d \k.
\end{equation}
Taking into account the spectral representation of the Riesz transforms, we may express their covariances as
\begin{align}
r_{gg}(\xib) &=  \iint \cos^2(\kappa) S_{ff}(\bk) e^{i \bk^T\xib} \, d \k , \\
r_{hh}(\xib) &=  \iint \sin^2(\kappa) S_{ff}(\bk) e^{i \bk^T\xib} \, d \k,
\end{align}
and we note that 
\begin{equation}
r_{ff}(\xib) = r_{gg}(\xib) + \, r_{hh}(\xib). \label{equ_rffrggrhh}
\end{equation}
This is an important property because it is a relationship between the three covariances, which we can use to simplify the statistical description of the monogenic signal. It can be shown that (\ref{equ_rffrggrhh}) only holds for stationary but not for nonstationary random fields. Additionally, we find that the cross-covariances are given by
\begin{align}
r_{fg}(\xib) &= - i  \iint \cos(\kappa) S_{ff}(\bk) e^{i \bk^T \xib} \, d \k , \\
r_{fh}(\xib) &= - i  \iint \sin(\kappa) S_{ff}(\bk) e^{i \bk^T \xib} \, d \k, \\
r_{gh}(\xib) &= \frac{1}{2} \iint \sin(2 \kappa) S_{ff}(\bk) e^{i \bk^T \xib} \, d \k,
\end{align}
which satisfy 
\begin{align}
r_{gf}(\xib) = -r_{fg}(\xib), & & r_{hf}(\xib) = -r_{fh}(\xib), & &  r_{hg}(\xib) = r_{gh}(\xib).
\end{align}
This provides a complete second-order statistical characterization of the monogenic signal with six auto- and cross-covariances.

So far we have characterized the monogenic signal in terms of its (real) components, but now we shall obtain the quaternion-valued characterization. Section \ref{sec:rev_quat} has shown that in general a covariance and three complementary covariances, all quaternion-valued, are necessary to completely characterize a random quaternion. However, since the monogenic signal does not have a $k$-part, it is obvious that two complementary covariances suffice. As we will see shortly, for stationary random fields we only need one complementary covariance due to (\ref{equ_rffrggrhh}).

Let us start with the covariance of the monogenic signal, given by
\begin{equation}
r_{m m}(\xib) = \cov\{m(\x), m(\x - \xib)\} = 2 r_{f f}(\xib) - 2 i r_{f g}(\xib) - 2 j r_{f h}(\xib).
\end{equation}
This covariance function only specifies three of the six real covariances. In order to access the remaining three real covariances, one may be tempted to use the \emph{standard} complementary covariance. However, it is straightforward to show that $r_{m m^{\ast}}(\xib) = \cov\{m(\x), m^{\ast}(\x - \xib)\} = 0$ for all stationary random fields, so this complementary covariance does not provide any useful information. Consider instead the covariance between $m(\x)$ and $m^{(i)}(\x)$
\begin{equation}
r_{m m^{(i)}}(\xib) = 2 r_{g g}(\xib) - 2 i r_{f g}(\xib) + 2 k r_{g h}(\xib),
\end{equation}
which specifies a further two real covariances. Recalling now that $r_{h h}(\xib) = r_{f f}(\xib) - r_{g g}(\xib)$ gives us access to the final remaining real covariance, $r_{m m}(\xib)$ and $r_{m m^{(i)}}(\xib)$ together contain the same information as the six real covariance functions. We remark that instead of $r_{m m^{(i)}}(\xib)$, one can also use \cite{Via2010}
\begin{equation}
r_{m m^{(\eta)}}(\xib) = \cov\{m(\x), m^{(\eta)}(\x - \xib)\},
\end{equation}
where $\eta$ is any pure unit quaternion with zero $k$-part. 
The additional freedom of being able to choose an arbitrary basis provides flexibility that we shall use later.

\section{Isotropy and directionality}
\label{sec:iso-direct}

A fundamental characteristic of a random field is the degree of rotational invariance of its second-order statistical properties \cite[p.~57]{Christakos1992}. A spatially {\em isotropic} random field exhibits perfect rotational invariance as its covariance displays circular contour levels. Fields that are not isotropic are called {\em anisotropic}. A common subclass is the class of {\em geometrically anisotropic} fields \cite[p.~61]{Christakos1992} that have covariances with elliptical rather than circular contours. Taking this idea to the extreme, we arrive at {\em unidirectional} random fields, where there exists a rotation so that there is variation only in one of the two axes. Figure \ref{fig:random_fields} shows samples of three random images: isotropic, geometrically anisotropic, and unidirectional. In this section we provide formal definitions of these three types of random fields.


\subsection{Isotropy}

\begin{definition}
A second-order stationary random field $f(\x)$ is isotropic if the covariance of the field is finite and only depends on the magnitude of the lag, that is if $r_{f f}(\xib)=C_{I}\left(\sqrt{\xib^T\xib}\right)$ \cite{Stein1999}.
\end{definition}

\begin{figure}
        \centering
        \begin{subfigure}[b]{0.3\textwidth}
                \centering
                \includegraphics[width=\textwidth]{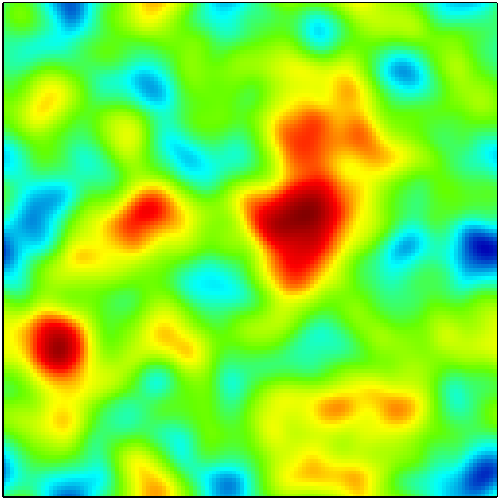}
                \caption{Isotropic}
                \label{fig:iso}
        \end{subfigure}%
        ~ 
        \begin{subfigure}[b]{0.3\textwidth}
                \centering
                \includegraphics[width=\textwidth]{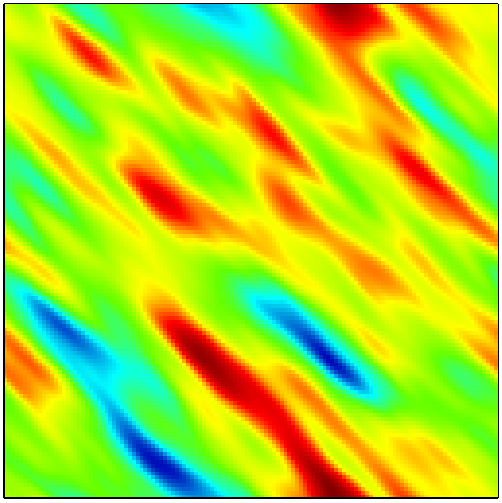}
                \caption{Anisotropic}
                \label{fig:aniso}
        \end{subfigure}
        ~ 
        \begin{subfigure}[b]{0.3\textwidth}
                \centering
                \includegraphics[width=\textwidth]{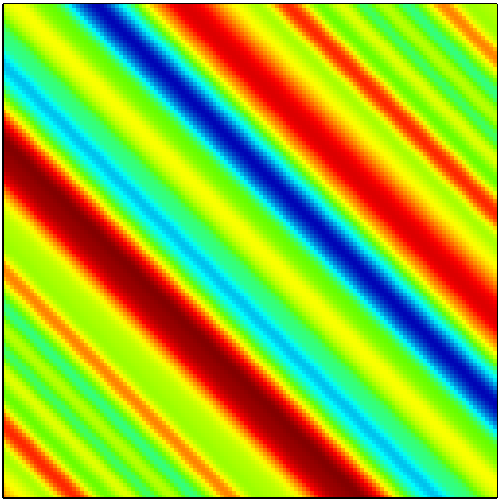}
                \caption{Unidirectional}
                \label{fig:uni}
        \end{subfigure}
        \caption{Sample images of random fields}
        \label{fig:random_fields}
\end{figure}

Figure \ref{fig:iso} shows a sample of an isotropic random field. We observe that it does not show any preferential treatment of any given image orientation. That is, if we rotate an isotropic random field, it retains its statistical properties for arbitrary rotation angle. We now derive the covariances of the monogenic signal for an isotropic random field. Let the spatial lag be expressed as $\xib =  \xi [\cos(\theta), \sin(\theta)]^T$, 
and let $S_{ff}(\k) = S_{ff}(k)$ be the power spectral density of $f(\x)$. Hence, $r_{f f}(\xib)$ is given by
\begin{align}
r_{f f}(\xib) &=  \iint S_{ff}(\bk) e^{i \bk^T\xib} \, d \k = \int_{0}^{\infty} \int_{-\pi}^{ \pi} S_{ff}(k) e^{i \bk^T\xib} \, k dk d \kappa \nonumber \\
&= 2 \pi \int_{0}^{\infty} S_{ff}(k) J_0(k \xi)  \, k dk = 2 \pi \mathcal{H}_{0}^{-1}(S_{ff}, \xi),
\end{align}
where we have used polar coordinates to solve the integral, $J_{\alpha}(\cdot)$ denotes the Bessel function of the first kind and order $\alpha$, and $\mathcal{H}_{\alpha}^{-1}(S_{ff}, \cdot )$ is the inverse Hankel transform of $S_{ff}$ and order $\alpha$ \cite{sneddon1979use}. Following similar derivations, the covariances of the Riesz transforms become
\begin{align}
r_{g g}(\xib) &= 2 \pi \cos^2(\theta) \mathcal{H}_{0}^{-1}(S_{ff}, \xi) - \frac{\cos(2 \theta)}{\xi} \mathcal{H}_{1}^{-1}(k^{-1} S_{ff}, \xi), \\
r_{h h}(\xib) &= 2 \pi \sin^2(\theta) \mathcal{H}_{0}^{-1}(S_{ff}, \xi) + \frac{\cos(2 \theta)}{\xi} \mathcal{H}_{1}^{-1}(k^{-1} S_{ff}, \xi),
\end{align}
and, as expected, they fulfill $r_{f f}(\xib) = r_{g g}(\xib) + r_{h h}(\xib)$. Finally, the cross-covariances are given by
\begin{align}
r_{f g}(\xib) &= - 2 \pi \cos(\theta) \mathcal{H}_{1}^{-1}(S_{ff}, \xi), \\
r_{f h}(\xib) &= - 2 \pi \sin(\theta) \mathcal{H}_{1}^{-1}(S_{ff}, \xi), \\
r_{g h}(\xib) &= \sin(2 \theta) \left[ \pi \mathcal{H}_{0}^{-1}( S_{ff}, \xi) - \frac{1}{ \xi} \mathcal{H}_{1}^{-1}(k^{-1} S_{ff}, \xi)\right].
\end{align}
From the previous equations, it is clear that the six real covariances are functions of only three quantities, namely $\mathcal{H}_{0}^{-1}( S_{ff}, \cdot)$, $\mathcal{H}_{1}^{-1}(S_{ff}, \cdot)$ and $\mathcal{H}_{1}^{-1}(k^{-1} S_{ff}, \cdot)$. To obtain these forms we used simplifying relationships between $\mathcal{H}_{0}^{-1}(\cdot)$, $\mathcal{H}_{1}^{-1}(\cdot)$ and $\mathcal{H}_{2}^{-1}(\cdot)$.
Therefore, we may characterize the second-order statistics of a monogenic signal for an isotropic random field using one properly selected quaternion-valued covariance, which gives us access to the three unknown quantities. We can, for instance, use $r_{m m^{(i)}}(\xib)$.

\subsection{Geometric anisotropy}

\begin{definition}
A second-order stationary random field $f(\x)$ is geometrically anisotropic if the covariance of the field is finite and only depends on the magnitude of the deformed lag as $r_{f f}(\xib)=C_{A}(\sqrt{\xib^T\mathbf{D}\xib})$,
for some $2\times 2$ symmetric positive definite matrix $\mathbf{D}$ that satisfies $\text{det}(\D) = 1$.
\end{definition}
Figure \ref{fig:aniso} shows a sample of an anisotropic random field. We can see that there is a preferred direction, which is clearly visible as we have chosen a matrix $\D$ with large condition number. Now we will investigate the statistical properties of the monogenic signal for a geometrically anisotropic random field. The covariance of $f(\x)$ is given by
\begin{equation}
r_{ff}(\xib) = 2 \pi \mathcal{H}_0^{-1}(S_{ff}, \tilde{\xi}),
\end{equation}
where $\tilde{\xi} = \sqrt{\sigma_1 \xi_1^2 + \sigma^{-1}_1 \xi_2^2}$, with $\sigma_1$ being the largest eigenvalue of $\D$. The proof can be found in Appendix \ref{sec:proof_cov_ani}. Contrary to isotropic random fields, we generally need two quaternion-valued covariances to completely characterize the statistics of a geometrically anisotropic random field. To see this, let us consider the covariance of the first Riesz transform, given by
\begin{equation}
r_{gg}(\xib)  = \sum_{l = -\infty}^{\infty} e^{i l (\tilde{\xi} + \pi/2)} a_l(\sigma_1,\alpha) \mathcal{H}_l ( S_{ff}, \tilde{\xi}), 
\end{equation}
where $a_l(\sigma_1,\alpha)$ are the Fourier coefficients of
\begin{equation}
 \beta(\sigma_1,\alpha) =\frac{\cos(\alpha) \cos(\kappa) \sqrt{\sigma_1^{2} \cos^2(\kappa) + \sin^2(\kappa) } \mp \sin(\alpha) \sin(\kappa) \sqrt{ \cos^2(\kappa) + \sigma_1^{-2} \sin^2(\kappa)}}{ \sigma_1 \cos^2(\kappa) + \sigma_1^{-1} \sin^2(\kappa)},
 \end{equation}
$\alpha$ is the angle of the dominant eigenvector of $\D$, and the sign depends on the determinant of the eigenvector matrix of $\D$ (i.e., whether it is a rotation or a reflection matrix). The proof is also presented in Appendix \ref{sec:proof_cov_ani}. Hence, just considering $r_{gg}(\xib)$, we see that at least two quaternion-valued covariance functions are required. One could argue that some (or almost all) of the coefficients $a_l(\sigma_1,\alpha)$ might be zero and, therefore, the correlation would depend only on a few $\mathcal{H}_l ( S_{ff}, \tilde{\xi})$, as in the isotropic case. To show that this is not true, Figure \ref{fig:betas} depicts the function $\beta(\sigma_1,\alpha)$ for $\sigma_1 = 0.5$, and $\alpha = 0$ and $\alpha = \pi/2$. We can see that several Fourier coefficients are needed to express $\beta(\sigma_1,\alpha)$ even for this simple case. For arbitrary values of $\sigma_1$ and $\alpha$, simulations have shown that there does not exist a simple Fourier expansion of the function.

\begin{figure}[t]
\centering
\includegraphics{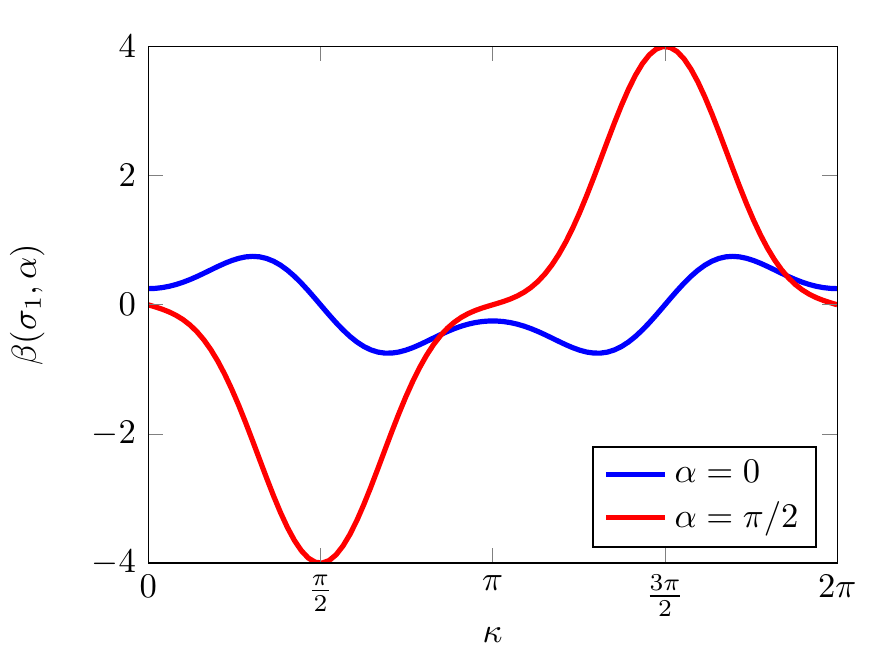}
\caption{Plot of $\beta(\sigma_1,\alpha)$ for $\sigma_1 = 0.5$, and $\alpha = 0$ and $\alpha = \pi/2$.}
\label{fig:betas}
\end{figure}

\subsection{Unidirectionality}

As the condition number of $\mathbf{D}$ in the geometric anisotropic model increases, the random field appears more and more directional. However, no matter what finite condition number is selected, the random field is still an aggregation of more than one wavenumber coefficient. We therefore introduce an idealized third class of random field, namely {\em purely directional} signals. 

\begin{definition}
A second-order stationary random field $f(\x)$ is purely unidirectional if its covariance takes the form $r_{f f}(\xib)=C_U(\n^T \xib)$, where $\n = [n_1,n_2]^T = [\cos(\nu), \sin(\nu)]^T$ is a unit-norm vector in the direction of the random field.
\end{definition}
A realization of a unidirectional random field is shown in Figure \ref{fig:uni}. Using the spectral representation, a unidirectional random field may be expressed as 
\begin{equation}
f(\x)=  \int d \mathcal{Z}_f(k') e^{i k' \bn^T \x},
\end{equation}
where $k'$ is a signed version of the wavenumber $k$, and $\cov\{d \mathcal{Z}_f(k'),d \mathcal{Z}_f(k'')\}=S_{ff}(k') \delta(k' -k'')\,dk' \, dk''$. As we can see, the unidirectional signal is thus a one-dimensional signal embedded in two dimensions \cite{Krieger1996,felsberg2009continuous,knutsson1994signal}.
The spectral representation allows us to write the Riesz transforms as
\begin{align}
\label{def:gh}
g(\x) &= \cos(\nu) s(\x), &
h(\x) &= \sin(\nu) s(\x),
\end{align}
where
\begin{align}
s(\x) = -i \int_{-\infty}^{\infty} \mathrm{sgn}(k') d \mathcal{Z}_f(k') e^{i k' \bn^T \x}
\end{align}
is a partial Hilbert transform in direction $\n$ \cite{Felsberg2001}. It is easy to show that the covariances are 
\begin{align}
r_{g g}(\xib) &= \cos^2(\nu) r_{f f}(\xib), &
r_{h h}(\xib) &= \sin^2(\nu) r_{f f}(\xib).
\end{align}
The cross-covariances are
\begin{align}
r_{f h}(\xib) &= \tan(\nu) r_{f g}(\xib), &
r_{g h}(\xib) &= \frac{1}{2} \sin(2 \nu) r_{f f}(\xib),
\end{align}
where
\begin{equation}
r_{f g}(\xib) = i \cos(\nu) \int_{-\infty}^{\infty} \mathrm{sgn}(k') S_{f f}(k') e^{i k' \bn^T \xib} d k.
\end{equation}
Similar to the case of isotropic random fields, we do not need five real covariances to characterize a unidirectional random field, but only two. Thus, only one quaternion-valued covariance suffices, for instance $r_{m m}(\xib)$. A consequence of this correlation structure is that unidirectional random fields are $\mathbb{C}^{\eta}$-proper:
\begin{theorem}[Characterization of a unidirectional random field]
\label{th:c_eta_uni}
The monogenic signal of a random field $f(\x)$ is $\mathbb{C}^{\eta}$-proper, where $\eta = \cos(\nu) i + \sin(\nu) j$, if and only if $f(\x)$ is stationary and unidirectional with direction $\n = [n_1,n_2]^T = [\cos(\nu), \sin(\nu)]^T$.
\end{theorem}
\begin{proof}
See Appendix \ref{uni_proof}.
\end{proof}
This result shows an analogy between the monogenic signal of a stationary unidirectional random field and the analytic signal of a stationary process, which is complex proper \cite{book_peter}. That is, the complementary covariance of the analytic signal $x_+(t)$ for a stationary complex process $x(t)$ is zero, i.e., $\cov\{x_+(t), x^{\ast}_+(t - \tau)\} = 0, \forall \tau$.

\section{Test for unidirectionality}
\label{sec:test}

In this section, we propose a measure for the degree of unidirectionality in a random field and estimate its preferred direction. We then use this measure to build a test for whether there is sufficient statistical evidence to classify a random field as unidirectional.

\subsection{Measure of unidirectionality}
It follows directly from~\eqref{def:gh} that the monogenic signal of a unidirectional random field may be expressed as
\begin{equation}
m(\x) = f(\x) + \eta s(\x),
\end{equation}
where $s(\x)$ is the partial Hilbert transform in the direction $\eta$ of the field. This means that $m^{(\eta)}(\x) = m(\x)$ for the unit quaternion $\eta = \cos(\nu)i + \sin(\nu)j$. Therefore, the quantity
\begin{equation}
\label{eq:measure}
\mathop{\min}_{\eta} E\left[ \frac{1}{2} \left| m(\x) - m^{(\eta)}(\x) \right|^2\right]
\end{equation}
is zero if $f(\x)$ is unidirectional and greater than zero if it is not. To see the validity of this statement for any stationary random field, we shall outline some simplifications. We may rewrite~\eqref{eq:measure}
as
\begin{equation}
\min_{\eta} E\left[ \frac{1}{2} \left|m(\x) - m^{(\eta)}(\x)\right|^2 \right] = \min_{\eta} E\left[ \frac{1}{2}|m(\x)|^2 + \frac{1}{2} |m^{(\eta)}(\x)|^2 - \text{Re}\left(m(\x)  m^{(\eta)^{\ast}}(\x)\right)\right].
\end{equation}
Taking into account that $|m^{(\eta)}(\x)| = |m(\x)|$ the previous expression equates to
\begin{equation}
\min_{\eta} E\left[ \frac{1}{2} \left|m(\x) - m^{(\eta)}(\x)\right|^2 \right] = \min_{\eta} E\left[ |m(\x)|^2 - \text{Re}\left(m(\x)  m^{(\eta)^{\ast}}(\x)\right)\right].
\end{equation}
Now, taking into account the definition of the covariances for a stationary random field, we find that 
\begin{equation}
\min_{\eta} E\left[ \frac{1}{2} \left|m(\x) - m^{(\eta)}(\x)\right|^2 \right] = \min_{\eta} \left[ r_{mm}(\0) - \text{Re}\left(r_{mm^{(\eta)}}(\0)\right)\right] = r_{mm}(\0) - \max_{\eta} \text{Re}\left(r_{mm^{(\eta)}}(\0)\right).
\end{equation}
Based on this mean-squared error, we propose to use
\begin{equation}
\label{eq:measure2}
\mathcal{U} = \frac{2 \displaystyle\mathop{\max}_{\eta} \text{Re}\left(r_{m m^{(\eta)}}(\0)\right)}{r_{m m}(\0)} - 1
\end{equation}
as a measure of unidirectionality. We will see momentarily that this measure is normalized to values between $0$ and $1$. The maximum $\mathcal{U} = 1$ is attained if the field is unidirectional, i.e., $m^{(\eta)}(\x) = m(\x)$. We now simplify the expression \eqref{eq:measure2} for $\mathcal{U}$. We start by considering the numerator. Expressing the involution over $\eta$ in terms of the involutions over the canonical basis $\{i,j,k\}$, after some tedious algebra the numerator becomes
\begin{align}
\label{eq:tediousR} \mathop{\max}_{\eta} \text{Re}\left(r_{m m^{(\eta)}}(\0)\right) &= \mathop{\max}_{\n, \|\n\|=1} \, \n^T \R \n, \\
\mbox{with}\quad \R &= \begin{bmatrix}  \text{Re} \left(r_{m m^{(i)}}(\0)\right) & \text{Im}_{k}\left(r_{m m^{(i)}}(\0)\right) \\
-\text{Im}_{k}\left(r_{m m^{(j)}}(\0)\right) & \text{Re}\left(r_{m m^{(j)}}(\0)\right)
\end{bmatrix}, \nonumber  
\end{align}
where $\text{Im}_{k}\left(q\right)$ is the $k$-component of the quaternion $q$. Hence, it is easy to show that \eqref{eq:measure2} can be written as
\begin{equation}
\label{eq:measure3}
\mathcal{U} = \frac{2 \lambda_{\mathrm{MAX}}(\R)}{r_{m m}(\0)} - 1,
\end{equation}
where $\lambda_{\mathrm{MAX}}(\R)$ is the largest eigenvalue of $\R$ and the direction $\n$ is given by the dominant eigenvector. Taking into account that the matrix $\R$ is given by
\begin{equation}
\R =  \begin{bmatrix} r_{f f}(\0)  + r_{g g}(\0)  -  r_{h h}(\0)  & 2 r_{g h}(\0) \\
2 r_{g h}(\0) & r_{f f}(\0)  - r_{g g}(\0)  +r_{h h}(\0)
\end{bmatrix},
\end{equation}
we use the closed-form expression for the largest eigenvalue of $\R$ to write
\begin{equation}
\mathcal{U} = \frac{\sqrt{ r^2_{g g}(\0) + r^2_{h h}(\0) - 2 r_{g g}(\0) r_{h h}(\0) + 4 r^2_{g h}(\0)}}{r_{f f}(\0) + r_{g g}(\0) + r_{h h}(\0)}.
\end{equation}
Using the properties of the covariances, $\mathcal{U}$ finally becomes
\begin{align}
\mathcal{U} = \frac{\sqrt{ (r_{g g}(\0) - r_{h h}(\0))^2 + 4 r^2_{g h}(\0)}}{2 r_{f f}(\0)},
\end{align}
which is similar to the coherency index measure introduced in \cite{multiresolution_monogenic} for deterministic images. It is now obvious that $\mathcal{U}$ is lower-bounded by $0$, so $\mathcal{U} \in [0, 1]$. To shed some light on $\mathcal{U}$, let us determine which signals minimize and maximize $\mathcal{U}$.

To achieve $\mathcal{U} = 0,$ we need $r_{g g}(\0) - r_{h h}(\0) = 0$ and $r_{g h}(\0) = 0$. This is equivalent to
\begin{align}
\iint \cos(2 \kappa) S_{ff}(\bk) \, d \k &= 0, &
\iint \sin(2 \kappa) S_{ff}(\bk) \, d \k &= 0.
\end{align}
Combining both conditions we have
\begin{equation}
\iint S_{ff}(\bk) \mathrm{e}^{- i 2 \kappa} \, d \k = 0,
\end{equation}
which may be rewritten as
\begin{equation}
\label{eq:a2}
\int_{0}^{\infty} a_2(k) k \, d k = 0,
\end{equation}
where
\begin{align}
a_2(k) =  \int_{-\pi}^{\pi} S_{ff}(\k)  \mathrm{e}^{- i 2 \kappa} \, d \kappa
\end{align}
is the second Fourier coefficient of $S_{ff}(\k) = S_{ff}(k,\kappa)$. For an isotropic random field, the power spectral density is $S_{ff}(\bk) = S_{ff}(k)$, yielding
\begin{equation}
\iint S_{ff}(\bk) \mathrm{e}^{- i 2 \kappa} \, d \k  = \int_{0}^{\infty} S_{ff}(k) \, k d k \int_{- \pi}^{\pi} \mathrm{e}^{- i 2 \kappa} \, d \kappa.
\end{equation}
It is clear that the integral in $\kappa$ is zero, and provided that $S_{ff}(k)$ satisfies
\begin{equation}
\int_{0}^{\infty} S_{ff}(k) \, k d k < \infty,
\end{equation}
the measure of unidirectionality for an isotropic random field is indeed zero. It may be tempting to think that isotropic random fields are the only fields that attain the lower bound of $0$. However, from the form of $e^{-i 2\kappa}$ we see that to achieve $\mathcal{U} = 0$ a sufficient condition is
\begin{equation}
S_{ff}\left(-k_1, k_2\right)=S_{ff}\left(k_1, k_2\right)
\end{equation}
and
\begin{equation}
S_{ff}\left(k_1, -k_2\right)=S_{ff}\left(k_1, k_2\right).
\end{equation}
Thus, any function that exhibits parity invariance independently in either of the two arguments, e.g., a separable function, also leads to $\mathcal{U} = 0$. One may thus argue that, in some sense, separable covariances are as far from being unidirectional as isotropic covariances.
While $\mathcal{U} = 0$ is only a necessary but not sufficient condition for a random field to be isotropic, $\mathcal{U} = 1$ is indeed necessary and sufficient for a field to be unidirectional. The proof of this statement follows along the lines of Appendix \ref{uni_proof}.

\subsection{Relationship with previously proposed measures for unidirectionality}
\label{sec:previous}

In this section we will review some related measures for the degree of unidirectionality (and also estimators of the direction), and state the similarities to and differences from our work. We first consider the \emph{Gaussian curvature} \cite{spivak1975comprehensive}. This is given by the determinant of the Hessian matrix, which contains all second-order partial derivatives. Assuming that the partial derivatives are continuous, the Gaussian curvature is
\begin{equation*}
\mathcal{GC}(\x) = \det \begin{bmatrix}
 f_{11}(\x) &  f_{12}(\x)  \\
 f_{12}(\x) & f_{22}(\x) 
\end{bmatrix}, \quad f_{l p}(\x) = \frac{\partial^2}{\partial x_l \partial x_p} f(\x), \ l,p = 1, 2,
\end{equation*}
which is a function of $\x$. It remains unclear how to obtain a single (global) and normalized measure for the degree of unidirectionality. Thus, in its current form, $\mathcal{GC}(\x)$ cannot be used to test for unidirectionality. 

A similar measure that employs partial derivatives uses a tensor-based estimator of the orientation \cite{DiClaudio:2010}. The main idea behind this approach is that the gradient of a unidirectional image is orthogonal to the direction, which is estimated as the eigenvector corresponding to the minimum eigenvalue of the matrix
\begin{equation*}
\S = \begin{bmatrix}
\int f_{1}^2(\x) d \x & \int f_{1}(\x) f_{2}(\x) d \x \\
\int f_{1}(\x) f_{2}(\x) d \x & \int f_{2}^2(\x) d \x
\end{bmatrix},
\end{equation*}
where $f_l(\x) = (\partial/\partial x_l) f(\x), \ l = 1,2,$ represents the partial derivative with respect to the $l$th dimension. At first glance, this estimate may seem similar to our proposed estimate of the direction, which is the eigenvector corresponding to the largest eigenvalue of the matrix $\R$ in (\ref{eq:tediousR}). However, the matrix $\R$ is composed of the covariances of the image and its Riesz transforms. Moreover, although one could imagine using the minimum eigenvalue of $\S$ as a measure for the degree of unidirectionality, it is not clear how to normalize this measure and determine the threshold for a test.

Finally, let us turn to the approach in \cite{multiresolution_monogenic}, which proposes a deterministic counterpart of our measure. In this approach, the direction is estimated as the direction that maximizes the integral of the partial Hilbert transform \cite{SP_computer_vision}. This idea boils down to finding the principal eigenvector of
\begin{equation*}
\J = \begin{bmatrix}
\int g^2(\x) d \x & \int g(\x) h(\x) d \x \\
\int g(\x) h(\x) d \x & \int h^2(\x) d \x
\end{bmatrix}.
\end{equation*}
Based on this matrix, \cite{multiresolution_monogenic} defines a degree of unidirectionality, termed the coherency index, as 
\begin{equation*}
\chi = \frac{\lambda_{\text{MAX}}(\J) - \lambda_{\text{MIN}}(\J)}{\lambda_{\text{MAX}}(\J) + \lambda_{\text{MIN}}(\J)},
\end{equation*}
which is bounded between $0$ and $1$. After substituting expressions for the eigenvalues, the coherency has a form similar to our measure. Nevertheless, there are a few important differences. First of all, while the definition in \cite{multiresolution_monogenic} may seem ad-hoc, we provide a clear interpretation as a normalized mean square error. Moreover, our measure is based on a stochastic formulation, which allows us to perform statistical analysis and to model the effect of finite sample sizes. Based on this analysis we will be able to derive a threshold to test for unidirectionality, which is the topic of the next section.

\subsection{Determining the threshold}
\label{sec:analysis}

In order to test whether there is statistical evidence to classify a random field as unidirectional, we need to determine a threshold for our measure of unidirectionality $\mathcal{U}$. This threshold needs to take into account that, in practice, we only observe a single finite patch of a random field and therefore work with the periodic discrete Riesz transform. Recall that $\widetilde{m}(\x)$ is the periodic monogenic signal. We may estimate our measure of unidirectionality from a given realization of the random field of size $N \times N$ as
\begin{equation}
\label{eq:measurehat}
\hat{\mathcal{U}} = \frac{2 \lambda_{\mathrm{MAX}}(\hat{\R})}{\hat{r}_{\widetilde{m} \widetilde{m}}(\0)} - 1,
\end{equation}
where
\begin{equation}
\hat{\R} =  \begin{bmatrix}  \text{Re} \left(\hat{r}_{\widetilde{m} \widetilde{m}^{(i)}}(\0)\right) & \text{Im}_{k}\left(\hat{r}_{\widetilde{m} \widetilde{m}^{(i)}}(\0)\right) \\
-\text{Im}_{k}\left(\hat{r}_{\widetilde{m} \widetilde{m}^{(j)}}(\0)\right) & \text{Re}\left(\hat{r}_{\widetilde{m} \widetilde{m}^{(j)}}(\0)\right)
\end{bmatrix},
\end{equation}
the cross-covariances are
\begin{equation}
\hat{r}_{\widetilde{m} \widetilde{m}^{(\diamond)}}(\0) = \frac{1}{N^2} \sum_{n, n'} \widetilde{m}(\x_{n, n'}) \widetilde{m}^{(\diamond)^{\ast}}(\x_{n, n'}),
\end{equation}
and $\widetilde{m}^{(\diamond)}(\x)$ stands for either $\widetilde{m}(\x)$, $\widetilde{m}^{(i)}(\x)$ or $\widetilde{m}^{(j)}(\x)$. Then, the estimated measure of unidirectionality may be expressed as
\begin{equation}
\label{eq:Uhat}
\hat{\mathcal{U}} = \frac{2 \lambda_{\mathrm{MAX}}(\hat{\R})}{\hat{r}_{\widetilde{m} \widetilde{m}}(\0)} - 1 = \frac{\sqrt{ \hat{r}^2_{\tilde{g} \tilde{g}}(\0) + \hat{r}^2_{\tilde{h} \tilde{h}}(\0) - 2 \hat{r}_{\tilde{g} \tilde{g}}(\0) \hat{r}_{\tilde{h} \tilde{h}}(\0) + 4 \hat{r}^2_{\tilde{g} \tilde{h}}(\0)}}{\hat{r}_{f f}(\0) + \hat{r}_{\tilde{g} \tilde{g}}(\0) + \hat{r}_{\tilde{h} \tilde{h}}(\0)},
\end{equation}
and the estimated preferred direction is given by the dominant eigenvector of $\hat{\R}$. In order to determine a threshold for $\hat{\mathcal{U}}$, we need to determine the probability of false alarm, which requires knowing the distribution of the statistic $\hat{\mathcal{U}}$ under the null hypothesis ``the random field is unidirectional''. As one may expect, deriving the distribution of $\hat{\mathcal{U}}$ is far from trivial. However, for a bandpass random field, we are able to establish the following approximation of the false alarm probability.
\begin{theorem}
\label{th_pfa}
Consider a finite realization of a random field of size $N \times N$ whose power spectral density is bandpass with lower cutoff frequency $\lambda_l$. Using the estimated measure of directionality $\hat{\mathcal{U}}$ to test whether the field is unidirectional, the threshold $1 - \eta$ is determined such that the false alarm probability satisfies
\begin{equation}
\label{eq:pfa}
P\left(\hat{\mathcal{U}} \leq 1 - \eta\right) \leq \frac{1}{N \eta}  \left[\frac{4}{ \pi^2} \frac{1}{\lambda_l} - \frac{4}{9} \lambda_l  +  \frac{1}{3} - \frac{4}{\pi^2} \right].
\end{equation}  
Moreover, for unidirectional random fields, $\hat{\mathcal{U}} \rightarrow 1$ as $N \rightarrow \infty$, and the error $1 - \hat{\mathcal{U}}$ that is due to considering a finite patch of the random field decays with $1/N$. 
\end{theorem}

Establishing these results is very involved and requires lengthy derivations. Readers that are not interested in the technical details may therefore wish to skip the remainder of this section and simply apply the algorithm presented as Alg. \ref{alg:detector}. Note that selecting the threshold (Line \ref{lag_line:threshold} in Alg. \ref{alg:detector}) may alternatively be done based on the estimated power spectral density.

\begin{algorithm}[!t]
\DontPrintSemicolon
\LinesNumberedHidden
\KwIn{Bound for the false alarm probability $\epsilon$, and lower cutoff frequency $\lambda_l$.}
\ShowLn Compute the periodic Riesz transforms $\tilde{g}(\x)$ and $\tilde{h}(\x)$ \; 
\ShowLn Use $f(\x)$, $\tilde{g}(\x)$ and $\tilde{h}(\x)$ to estimate the covariances $\hat{r}_{ff}(\0)$, $\hat{r}_{\tilde{g} \tilde{g}}(\0)$, $\hat{r}_{\tilde{h} \tilde{h}}(\0)$, and $\hat{r}_{\tilde{g} \tilde{h}}(\0)$ \;
\ShowLn Obtain $\hat{\mathcal{U}}$ using \eqref{eq:Uhat} \;
\ShowLn \label{lag_line:threshold} Determine $\eta$ such that the right hand side of \eqref{eq:pfa} equals $\epsilon$ \;
\If{$\hat{\mathcal{U}} \geq 1 - \eta$}{$f(\x)$ is unidirectional}
\Else{$f(\x)$ is not unidirectional}

\caption{Detector of unidirectionality using the monogenic signal.}
\label{alg:detector}
\end{algorithm}

\subsection{Derivations}\label{sec:deriv}

In order to establish Theorem \ref{th_pfa}, we proceed as follows. We first consider a deterministic plane wave and evaluate the error that is due to the finite size of the random field. We then generalize this result to a random plane wave, and finally to a unidirectional random field, which we write as an infinite sum of random plane waves. We begin with the following lemma. 

\begin{lemma}
\label{lem:det_plane_wave}
Considering a deterministic plane wave, given by
\begin{equation}
f(\x)=A \cos(k_0 \mathbf{n}^T \x+\phi),
\end{equation}
where $A, \phi,$ and $k_0 = 2 \pi \lambda_0$ are given numbers, we have
\begin{equation}
U_2 \triangleq  1 - \hat{\cal U}  = \sum_{\lambda_1=-N/2}^{N/2-1}\sum_{\lambda_2=-N/2}^{N/2-1} C(\lambda_1,\lambda_2) + \mathcal{O}(1/N^2),
\label{summand1}
\end{equation}
where
\begin{equation}
C(\lambda_1,\lambda_2) = \frac{2}{N^4} \left[\frac{(\lambda_1 n_2- \lambda_2 n_1)^2}{\lambda_1^2 + \lambda_2^2}\right] \frac{\sin^2(\pi N(\lambda_0 n_1-\lambda_1))}{\sin^2(\pi (\lambda_0 n_1-\lambda_1))} \frac{\sin^2(\pi N(\lambda_0 n_2-\lambda_2))}{\sin^2(\pi (\lambda_0 n_2-\lambda_2))},
\end{equation}
with $k_i = 2 \pi \lambda_i$.
\end{lemma}

\begin{IEEEproof}
See Appendix \ref{lem:det_plane_wave_proof}.
\end{IEEEproof}

Let us now analyze \eqref{summand1} more carefully. To do so, we decompose $\lambda_0 n_l N$ as $\lambda_0 n_l N =  \lfloor \lambda_0 n_l N \rceil + c_l$, where $\lfloor \lambda_0 n_l N \rceil$ and $c_l$ are the integer and fractional parts of $\lambda_0 n_l N$, respectively. Defining the transformed indices $j_l = k_l - \lfloor \lambda_0 n_l N \rceil$, we may write $C(\lambda_1,\lambda_2) = \tilde{C}(j_1,j_2)$. We can consider different regions for the orders of magnitude of the transformed indices and analyze $\tilde{C}(j_1,j_2)$. For small values of both indices, i.e., $j_l = \mathcal{O}(1)$, the function becomes $\tilde{C}(j_1,j_2) = \mathcal{O}(1/N^2)$. For large values of only one of the indices, i.e., $j_1 = \mathcal{O}(1)$ and $j_1 = \mathcal{O}(N)$ or vice-versa, we find that $\tilde{C}(j_1,j_2) = \mathcal{O}(1/N^2)$. Finally, for large values of both indices $j_l = \mathcal{O}(N)$, we have $\tilde{C}(j_1,j_2) = \mathcal{O}(1/N^{4})$. Now, defining the sets $\mathcal{J}_{\alpha \alpha'} = \{(j_1, j_2) : j_1 = \mathcal{O}(N^{\alpha}), j_2 = \mathcal{O}(N^{\alpha'})\}$, we may decompose $U_2$ as
\begin{equation}
U_2= {\mathcal{B}}_{01}+{\mathcal{B}}_{10} - {\mathcal{B}}_{00} + {\mathcal{B}}_{11},
\label{o1}
\end{equation}
where each term is given by
\begin{equation}
{\mathcal{B}}_{\alpha \alpha'} = \sum_{(j_1,j_2) \in {\mathcal{J}}_{\alpha \alpha'}} \tilde{C}(j_1,j_2).
\end{equation}
Finally, taking into account that the size of the sets $\mathcal{J}_{\alpha \alpha'}$ is $|\mathcal{J}_{\alpha \alpha'}| = N^{\alpha + \alpha'}$, it is easy to prove that ${\mathcal{B}}_{01} = \mathcal{O}(1/N)$ and ${\mathcal{B}}_{10} = \mathcal{O}(1/N)$, whereas ${\mathcal{B}}_{00} = \mathcal{O}(1/N^2)$ and ${\mathcal{B}}_{11} = \mathcal{O}(1/N^2)$. Hence, the error is 
\begin{equation}
U_2 = {\mathcal{B}}_{01}+{\mathcal{B}}_{10} + \mathcal{O}(1/N^2) = \mathcal{O}(1/N),
\end{equation}
that is, it decays linearly with $N$. Hence, from the previous analysis we find a simpler expression for $U_2$, which is given in the following lemma.
\begin{lemma}
\label{lem:det_plane_wave_2}
For a deterministic plane wave, the error due to the finite periodic discrete Riesz transform is
\begin{equation}
U_2 = \frac{2}{N}\left[ \sin^{2}(\pi c_2) {\cal G}(\lambda_0,n_1,n_2) + \sin^{2}(\pi c_1) {\cal G}(\lambda_0,n_2,n_1)\right]  + \mathcal{O}(1/N^2),
\end{equation}
where ${\cal G}(\lambda_0,n_1,n_2) = {\cal G}_+(\lambda_0,n_1,n_2)+{\cal G}_-(\lambda_0,n_1,n_2)$ and
\begin{equation}
\label{eq:integral_G}
{\cal G}_{\pm}(\lambda_0,n_1,n_2) = n_1^2 \int_{0}^{1/2} \frac{1}{\lambda_0^2 n_1^2+(\lambda_0 n_2 \pm \lambda)^2} \frac{\lambda^2}{\sin^2(\pi \lambda) }\,d \lambda.
\end{equation}
The convergence rate with which $\hat{\mathcal{U}} \rightarrow 1$ as $N \rightarrow \infty$ is therefore linear in $N$.
\end{lemma}

\begin{IEEEproof}
See Appendix \ref{lem:det_plane_wave_2_proof}.
\end{IEEEproof}

So far, we have only considered a deterministic plane wave. In the following, we will generalize our results to a random plane wave with random amplitude, phase, and direction but fixed frequency. If we would like to test the null hypothesis that a realization of a random field is a plane wave vs. the alternative that it is not unidirectional, we need to choose a threshold for $\hat{\mathcal{U}}$ (or $U_2 = 1 - \hat{\mathcal{U}}$). In order to determine the probability of false alarm, we would need to know the distribution of $U_2$ under the null hypothesis. As deriving this distribution seems extremely difficult, we instead obtain a conservative bound on the probability of false alarm using Markov's inequality \cite{gubnerprobability}. Concretely, Markov's inequality states that
\begin{equation}
P\left(U_2 \geq \eta\right) \leq \frac{E[U_2]}{\eta}
\end{equation}
for a positive random variable $U_2$. Hence, we have to find $E[U_2]$ for a random plane wave, which is presented in the following theorem.
\begin{theorem}
\label{th:U2_plane_wave}
The expectation of $U_2$ for a random plane wave is 
\begin{equation}
E[U_2] =  \frac{1}{N}\left( \frac{4}{ \pi^2} \frac{1}{\lambda_0} - \frac{4}{9} \lambda_0 +   \frac{1}{3} - \frac{4}{\pi^2} \right).
\end{equation}
\end{theorem}
\begin{IEEEproof}
See Appendix \ref{th:U2_plane_wave_proof}.
\end{IEEEproof}

Now, using Theorem \ref{th:U2_plane_wave} and Markov's inequality, we may bound the probability of $U_2$ exceeding $\eta$ as
\begin{equation}
P\left(U_2 \geq \eta\right) \leq \frac{1}{N \eta}\left( \frac{4}{ \pi^2} \frac{1}{\lambda_0} - \frac{4}{9} \lambda_0 +   \frac{1}{3} - \frac{4}{\pi^2} \right),
\end{equation}
which is only valid when the right-hand side is smaller than one. The following theorem generalizes this result to unidirectional random fields with arbitrary PSD $S(\lambda)$, where $\lambda = k/2 \pi$.
\begin{theorem}
\label{th:U2_unidirectional}
The expectation of $U_2$ for a unidirectional random field is 
\begin{equation}
E\left[U_2\right] 
=\frac{1}{N} \frac{ \displaystyle \int_0^{1/2} \left[ \frac{4}{ \pi^2} \frac{1}{\lambda} - \frac{4}{9} \lambda +   \frac{1}{3} - \frac{4}{\pi^2}\right] S(\lambda) \, d \lambda}{\displaystyle \int_0^{1/2} S(\lambda)  \,d \lambda} + \mathcal{O}(1/N^2).
\label{EU22}
\end{equation}
\end{theorem}
\begin{IEEEproof}
See Appendix \ref{th:U2_unidirectional_proof}.
\end{IEEEproof}

This shows that, for unidirectional random fields, the rate with which $E[\hat{\mathcal{U}}] \rightarrow 1$ as $N \rightarrow \infty$ is $1/N$, just like it was for deterministic plane waves. Unfortunately, $E[\hat{\mathcal{U}}]$ still depends on the power spectral density (PSD) of the random field. Since this quantity is generally unknown, we should bound $E[\hat{\mathcal{U}}]$ to obtain an expression that does not depend on the PSD. In the derivations for the random plane wave, we already required that the frequency of the plane wave was not too low or too high. We shall therefore assume that the random field is bandpass. 

\begin{lemma}
The expectation $E[U_2]$ may be bounded as
\begin{equation}
\label{bound_E_U2}
E[U_2] \leq \frac{1}{N}  \left[\frac{4}{ \pi^2} \frac{1}{\lambda_l} - \frac{4}{9} \lambda_l  +  \frac{1}{3} - \frac{4}{\pi^2} \right] + \mathcal{O}(1/N^2),
\end{equation}  
where $\lambda_l$ is the lowest frequency. Equality in \eqref{bound_E_U2} is attained if the PSD takes the form
\begin{eqnarray}
S(\lambda)=\sigma^2 \delta(\lambda - \lambda_l).
\end{eqnarray}
On the other hand, if the PSD has support at higher frequencies, then \eqref{bound_E_U2} is a strict inequality. 
\end{lemma}

\begin{IEEEproof}
Let us define
\begin{equation}
\tilde{S}(\lambda) = \frac{S(\lambda)}{\displaystyle \int_{\lambda_l}^{\lambda_h} S(\lambda) \, d \lambda} ,
\end{equation}
where $\lambda_l$ and $\lambda_h$ denote the lowest and highest frequency component of the PSD. Observing that $\tilde{S}(\lambda)$ integrates to one, $E[U_2]$ may be written as\footnote{For the sake of notational simplicity we ignore the term $\mathcal{O}(1/N^2)$.}
\begin{equation}
E\left[U_2\right] = \frac{1}{N}  \left[ \int_{\lambda_l}^{\lambda_h} \left( \frac{4}{ \pi^2} \frac{1}{\lambda} - \frac{4}{9} \lambda \right) \tilde{S}(\lambda) \, d \lambda +  \frac{1}{3} - \frac{4}{\pi^2} \right].
\end{equation}
Now consider the integral
\begin{equation}
\int_{\lambda_l}^{\lambda_h} \left( \frac{4}{ \pi^2} \frac{1}{\lambda} - \frac{4}{9} \lambda \right) \tilde{S}(\lambda) \, d \lambda.
\end{equation}
By taking into account that the function in parentheses is monotonically decreasing and positive on the considered interval, we may write
\begin{equation}
\int_{\lambda_l}^{\lambda_h} \left( \frac{4}{ \pi^2} \frac{1}{\lambda} - \frac{4}{9} \lambda \right) \tilde{S}(\lambda) \, d \lambda \leq \left( \frac{4}{ \pi^2} \frac{1}{\lambda_l} - \frac{4}{9} \lambda_l \right) \int_{\lambda_l}^{\lambda_h} \tilde{S}(\lambda) \, d \lambda = \frac{4}{ \pi^2} \frac{1}{\lambda_l} - \frac{4}{9} \lambda_l,
\end{equation}
and the proof follows.
\end{IEEEproof}

Theorem \ref{th_pfa} is now obtained by applying Markov's inequality.

\section{Numerical results}
\label{sec_numres}

In this section, we present simulation results illustrating the behavior of our measure of directionality, which we analyzed theoretically in the previous section. First, we examine the statistical behavior of our measure for different kinds of random fields. Then, we apply our measure to a real-world problem, where we detect unidirectional patches on the surface of Venus.

\subsection{Statistical behavior of $\mathcal{U}$}

We generated samples of isotropic, geometrically anisotropic, and purely unidirectional random fields. The PSD of the isotropic and geometrically anisotropic random fields is a \emph{shifted} (band-pass) Mat\'ern covariance function \cite{Stein1999}, given by
\begin{equation*}
S_{ff}(\lambdab)=\frac{\sigma^2 \Gamma(\nu+1) (4\nu)^{\nu}}{\pi\Gamma(\nu) (\pi\rho)^{2\nu} \left[4\nu/(\pi\rho)^2+ \left(\sqrt{\lambdab^T {\mathbf{D}} \lambdab} - \lambda_0\right)^2\right]^{\nu+1}},
\end{equation*}
where $\k = 2 \pi \lambdab$, $\Gamma(\cdot)$ denotes the Gamma-function, and ${\mathbf{D}}=\mathbf{I}_2$ for the isotropic field, and
\begin{equation*}
\D = \frac{1}{\sqrt{0.2775}} \begin{bmatrix}
1 & 0.85 \\ 0.85 & 1
\end{bmatrix}
\end{equation*}
for the anisotropic random field. The purely unidirectional random field has PSD
\begin{equation*}
S_{ff}(\lambda)=\frac{\sigma^2 \Gamma(\nu+1) (4\nu)^{\nu}}{\pi\Gamma(\nu) (\pi\rho)^{2\nu} \left[4\nu/(\pi\rho)^2+ \left(\lambda - \lambda_0\right)^2\right]^{\nu+1}}
\end{equation*}
along the line $\lambdab = \lambda \n$. For all fields, we chose $\nu= 1.5, \, \rho= 20$, $\sigma^2=1$ and $\lambda_0 = 0.1$ as the center frequency.

\begin{figure}[t]
\centering
\includegraphics{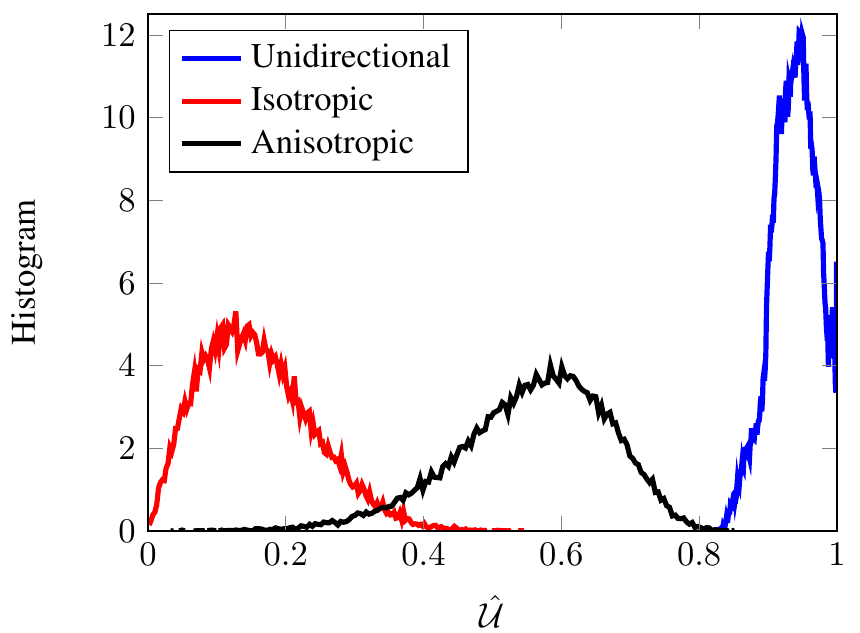}
\caption{Estimated PDF of the three different kinds of random fields with $N = 64$.}
\label{fig:histograms}
\end{figure}

We computed our estimated measure of unidirectionality $\hat{\mathcal{U}}$ and estimated the probability density functions (PDFs) of $\hat{\mathcal{U}}$ for the three different kinds of random fields of size $64 \times 64$ from $10,000$ realizations. The results are shown in Figure \ref{fig:histograms}. We can see that the isotropic random fields have $\hat{\mathcal{U}}$ close to zero, and the unidirectional random fields have $\hat{\mathcal{U}}$ close to one, in some cases even identically one. These cases happen when the direction is aligned with one of the coordinate axes, i.e., $\nu = 0$ or $\nu = \pi/2$. The PDFs of isotropic and unidirectional random fields are clearly separated and, while the PDFs of anisotropic and unidirectional random fields are closer, they are still easily distinguished. For larger values of $N$, the PDFs are even more clearly separated. As $N$ increases so does their separation.

Next, we would like to corroborate the results of Section \ref{sec:analysis}. Consider first a random plane wave, with fixed amplitude, random phase and direction. Figure \ref{fig:Exp_plane} shows the value of $E[\hat{\mathcal{U}}]$ predicted by Theorem \ref{th:U2_plane_wave}, and the estimate obtained from Monte Carlo simulations, as a function of the  frequency $\lambda_0$. We can see that there is good agreement except at low frequencies.  When $\lambda_0$ becomes small, e.g. when $\lambda_0={\cal O}(1/N)$ the higher-order error terms in Section~\ref{sec:deriv} become appreciable, see e.g. Lemma \ref{lem:integral_g}. The intuitive understanding of this is that when the image contains variation over spatial scales larger than the patch we have observed, this will inevitably produce problems with our estimation methods. We therefore assume we have observed a patch large enough to characterize the important structure of the random field. 

\begin{figure}[t]
\centering
\includegraphics{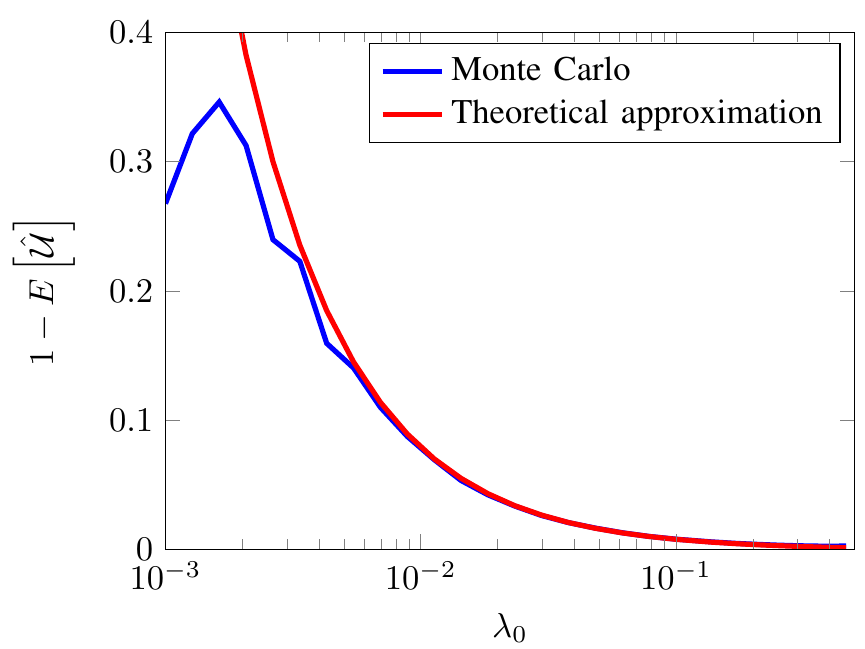}
\caption{Expectation of $\hat{\mathcal{U}}$ for a random plane wave.}
\label{fig:Exp_plane}
\end{figure}

Finally, we consider unidirectional random fields, which are generated as explained at the beginning of this section. Figure \ref{fig:Exp_uni} shows the expectation of $1 - \hat{\mathcal{U}}$, obtained through Monte Carlo simulations, the theoretical approximation given by Theorem \ref{th:U2_unidirectional}, and the upper bound given by \eqref{bound_E_U2}. For the upper bound, we have assumed that the lowest frequency component in the PSD is at $\lambda_l = 0.05$. We can see that the theoretical approximation is very good. While the bound is not very tight, it does display the same behavior with respect to the size $N$. That is, it decays with $1/N$, just like the Monte Carlo simulations. By using an estimated PSD in Theorem \ref{th:U2_unidirectional}, we obtain a tighter bound, which is shown in magenta.

\begin{figure}[t]
\centering
\includegraphics{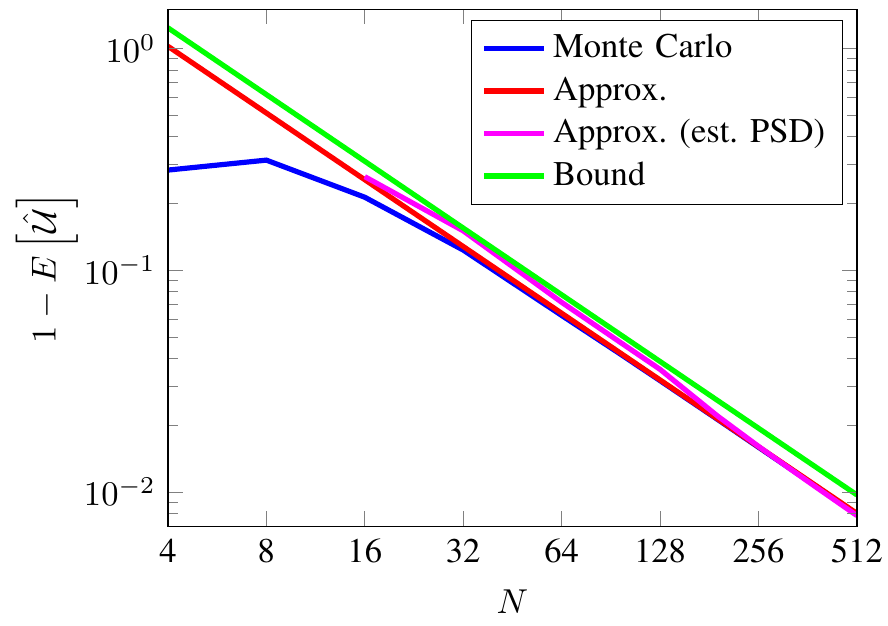}
\caption{Expectation of $1 - \hat{\mathcal{U}}$ for a unidirectional random field.}
\label{fig:Exp_uni}
\end{figure}

\subsection{Application}

We now present an application of our measure where we would like to detect unidirectional patches on the surface of the planet Venus. The image we consider is depicted in Figure \ref{fig:venus}. We would not expect the entire image to be unidirectional, but there are clearly some unidirectional patches. We therefore apply our measure of unidirectionality in a sliding-window fashion. For each pixel in the image, we calculate $\hat{\mathcal{U}}$ for a $16 \times 16$ neighborhood centered around that pixel. Three such neighborhoods are marked with a square in Figure \ref{fig:venus}, and shown in greater magnification in Figure \ref{fig:patches}. The arrows in Figure \ref{fig:venus} indicate the estimated direction of the three patches.

Looking at the three patches in Figure \ref{fig:patches} in detail, we see that Patch a (which corresponds to the bottom left square in Figure \ref{fig:venus}), is almost unidirectional and Patch c (which corresponds to the bottom right square), has a strong unidirectional component. Patch b, on the other hand, is the least unidirectional of the three patches. These observations match the estimated degrees of unidirectionality. The measure of unidirectionality can therefore be used to process large volumes of data automatically, and regions classified as unidirectional can then be scrutinized manually later. This is of strong interest, for instance, in the earth sciences, see e.g. \cite{simons2000isostatic}.

\begin{figure}[t]
\centering
\includegraphics[width=\columnwidth]{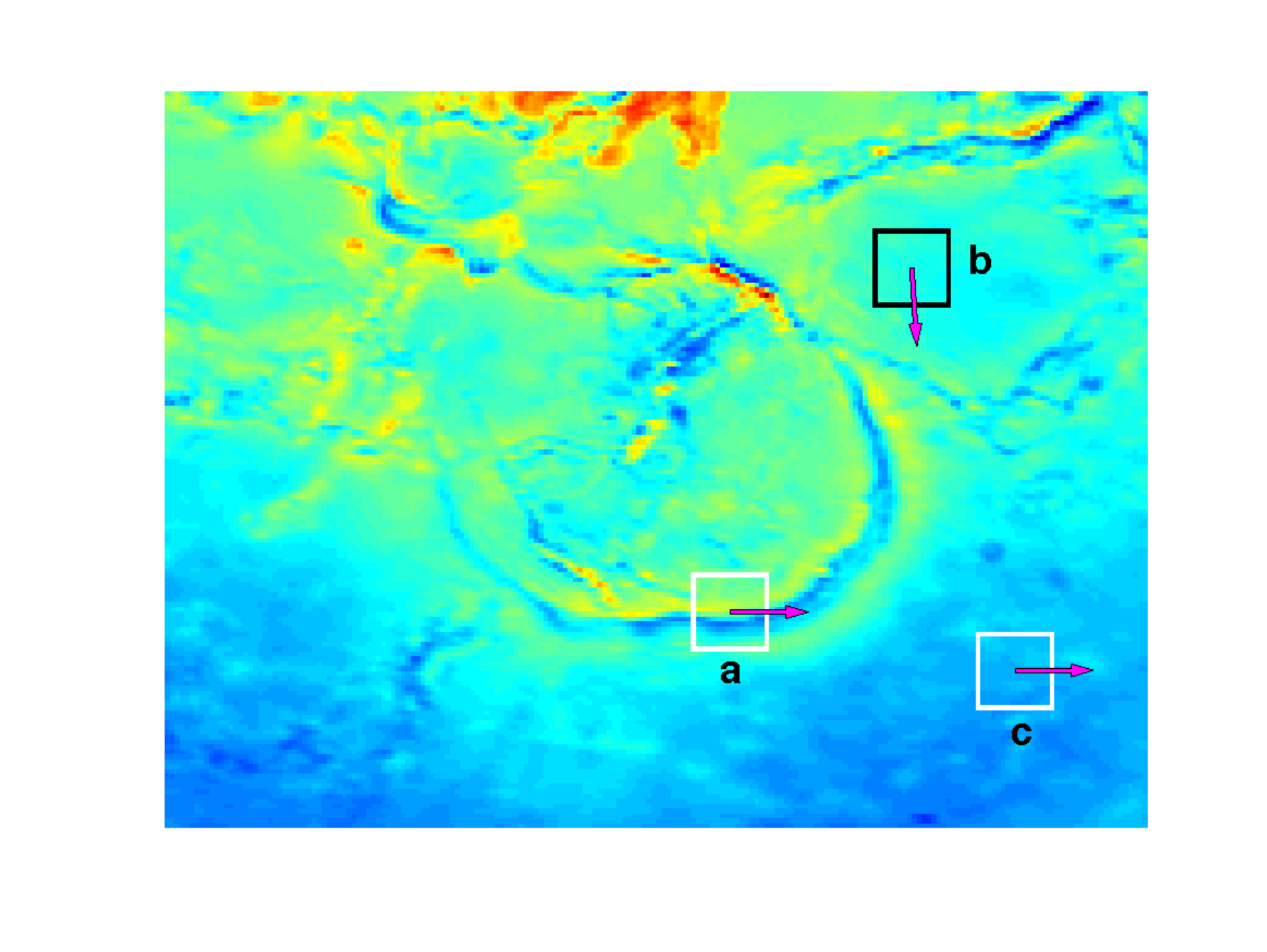}
\caption{Topography of Venus.}
\label{fig:venus}
\end{figure}

\begin{figure}
        \centering
        \begin{subfigure}[b]{0.3\textwidth}
                \centering
                \includegraphics[width=\textwidth]{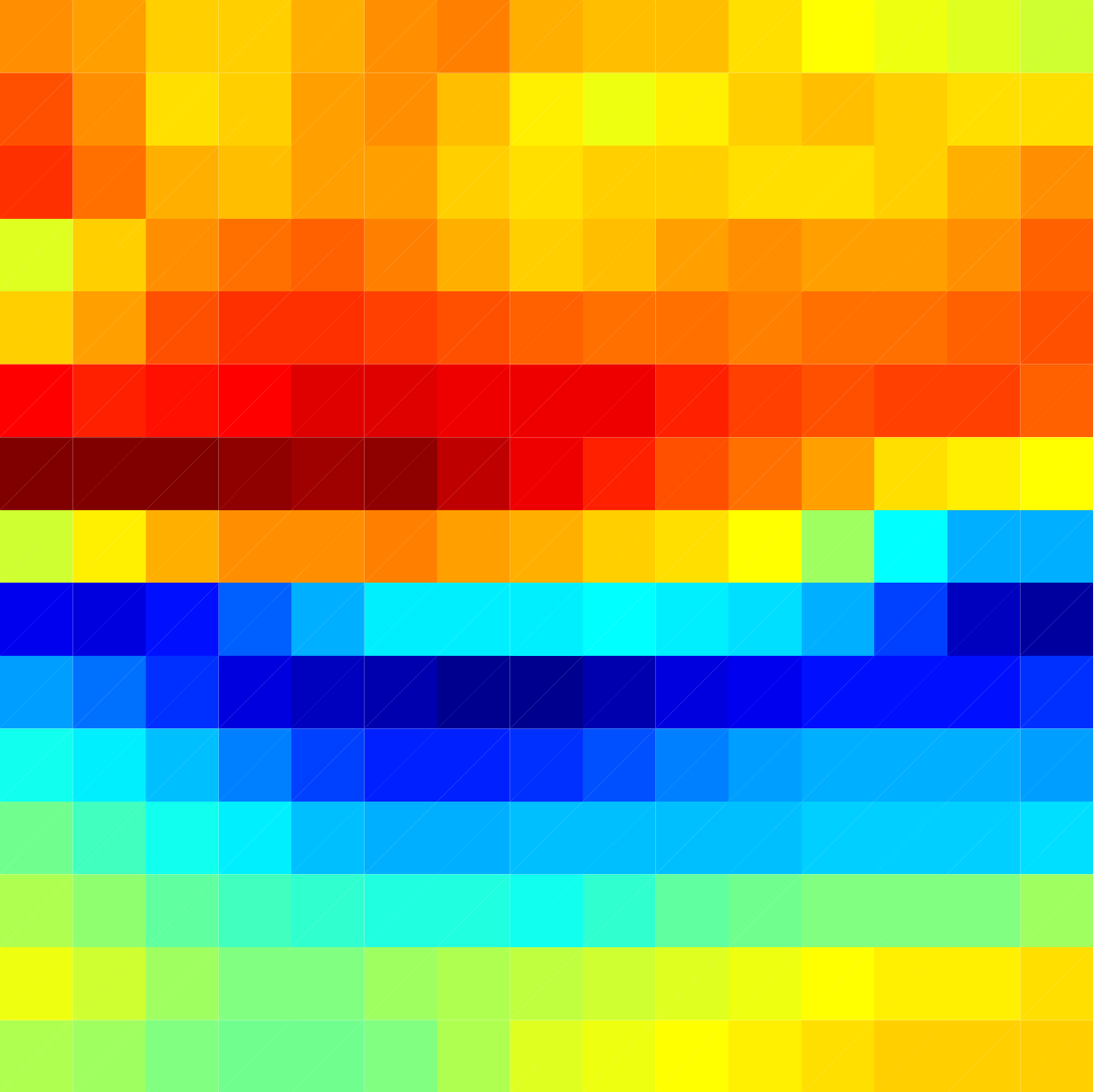}
                \caption{Patch a: $\hat{\mathcal{U}} = 0.906$}
                \label{fig:patch1}
        \end{subfigure}%
        ~ 
        \begin{subfigure}[b]{0.3\textwidth}
                \centering
                \includegraphics[width=\textwidth]{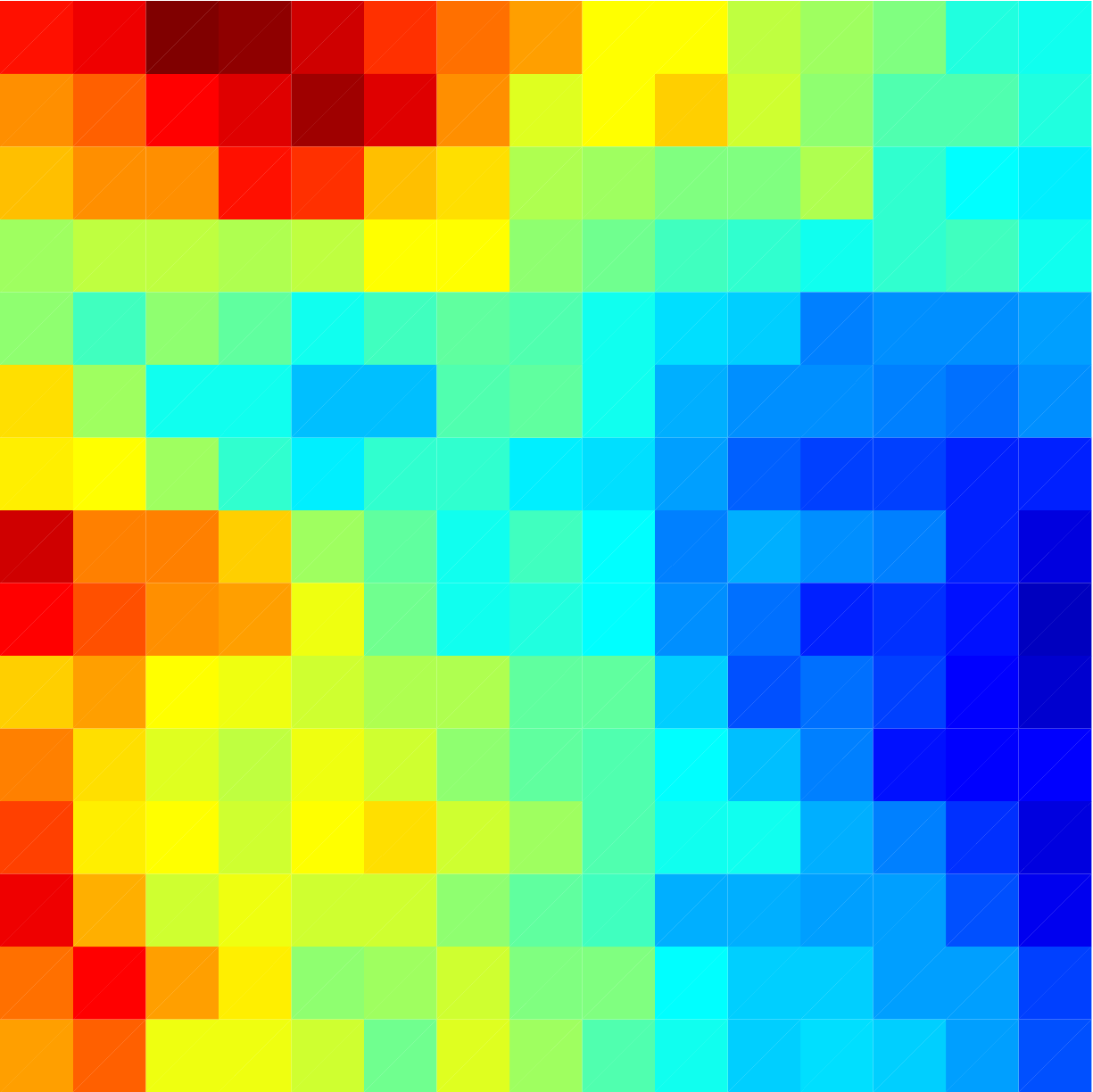}
                \caption{Patch b: $\hat{\mathcal{U}} = 0.412$}
                \label{fig:patch2}
        \end{subfigure}
        ~ 
        \begin{subfigure}[b]{0.3\textwidth}
                \centering
                \includegraphics[width=\textwidth]{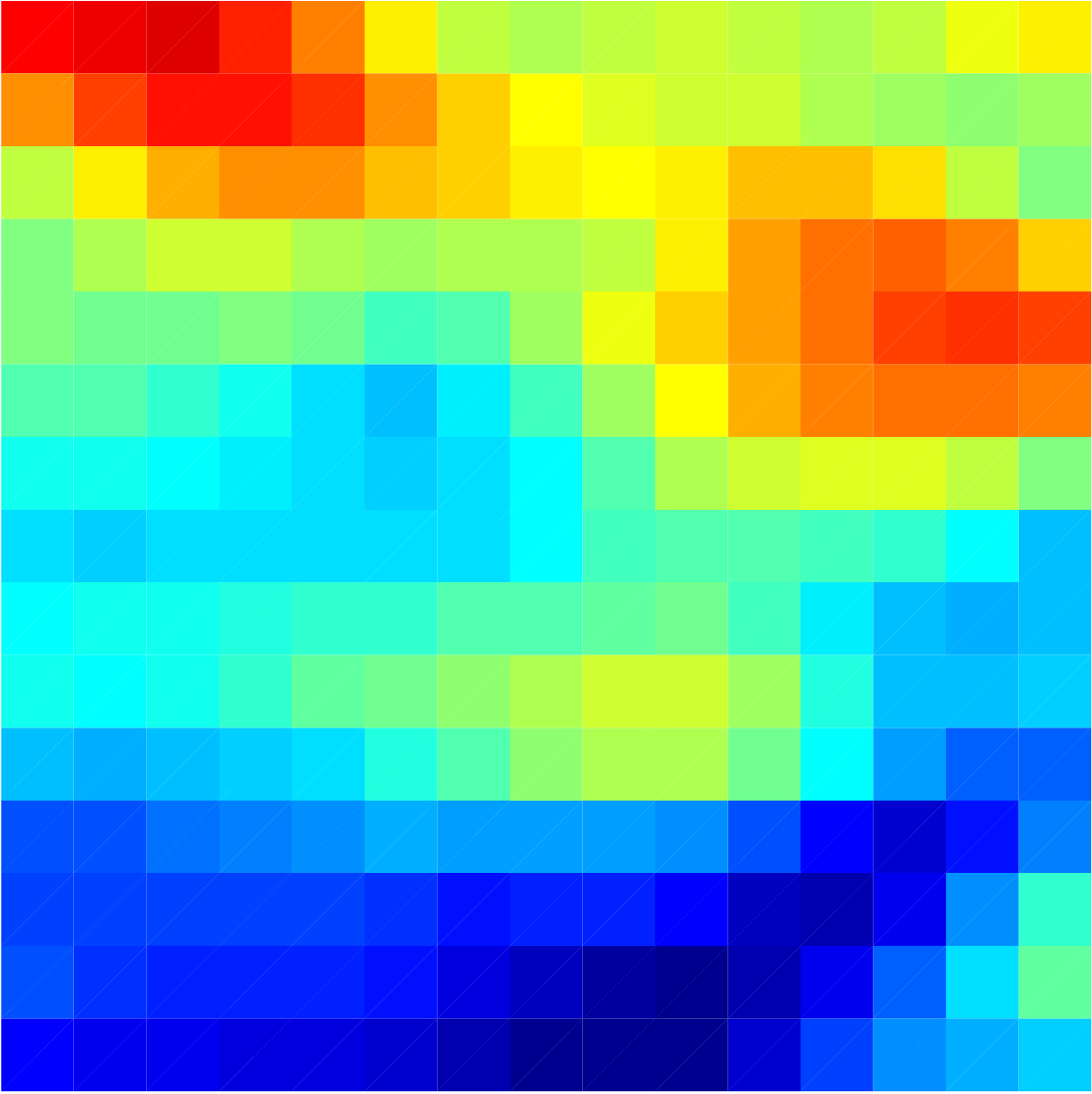}
                \caption{Patch c: $\hat{\mathcal{U}} = 0.784$}
                \label{fig:patch3}
        \end{subfigure}
        \caption{Several patches of Venus, as outlined in Figure \ref{fig:venus}.}
        \label{fig:patches}
\end{figure}

\section{Conclusions}

In this paper, we have introduced the random monogenic signal as a means to detect and analyze directional structures in images. Our main contribution compared to previous approaches is the statistical analysis of a measure of unidirectionality, which may be used to test an image for the presence of unidirectional components. Such a statistical test enables the {\em automatic} processing and classification of large volumes of data, which is of importance, for instance, in the earth sciences.

There are two main possible extensions of our work. First, we were only able to derive a rather loose bound on the detection threshold, so a tighter bound would be desirable. Second, directional structures are usually limited to smaller patches, and random fields may exhibit directionality only in certain frequency ranges in the wavenumber domain. Thus, a space- and wavenumber-localized version of our test should be the focus of future research.

\appendices

\section{Derivation of the covariances of a geometrically anisotropic random field}
\label{sec:proof_cov_ani}

The covariance function of a geometrically anisotropic random field is $r_{ff}(\xib) = C_A(\sqrt{\xib^T \D \xib})$, so its power spectral density may be expressed as $S_{ff}(\k) = S_{ff}\left(\sqrt{\k^T \D^{-1} \k}\right)$. Let us find an alternative expression for 
\begin{equation}
r_{ff}(\xib) =  \iint S_{ff} \left(\sqrt{\k^T \D^{-1} \k}\right) e^{i \k^T \xib} d \k.
\end{equation}
Using a square root matrix of $\D$, which satisfies $\D = \D^{1/2} \D^{T/2}$, we introduce the change of variables
\begin{equation}
\k' = \D^{-1/2} \k, \Rightarrow \k = \D^{1/2} \k', \Rightarrow d \k =| \mathrm{det}(\D^{1/2})| d \k' = d \k' ,
\end{equation}
where $\D^{-1/2} = \Sigmab^{-1/2} \U^T$, with $\Sigmab$ and $\U$ the matrix of eigenvalues and eigenvectors of $\D$, respectively. Therefore, the covariance becomes
\begin{equation}
r_{ff}(\xib) = \iint S_{ff} \left(\| \k' \|\right) e^{i \k'^T \D^{T/2} \xib} d \k'.
\end{equation}
Letting $\tilde{\xib} = \D^{T/2} \xib = \U \Sigmab^{1/2} \xib  = [\tilde{\xi} \cos(\tilde{\theta}), \tilde{\xi} \sin(\tilde{\theta})]^T$ and $\k' = [k' \cos(\kappa'), k' \sin(\kappa')]^T$, it is possible to write $r_{ff}(\xib)$ as
\begin{equation}
r_{ff}(\xib) =  \int_{0}^{\infty} \int_{0}^{2 \pi} S_{ff} \left( k' \right) e^{i k' \tilde{\xi} \cos(\kappa' - \tilde{\theta})} k' d k' d \kappa' = 2 \pi \mathcal{H}_0^{-1}(S_{ff}, \tilde{\xi}).
\end{equation}

The covariance of $g(\x)$ is given by
\begin{equation}
r_{gg}(\xib) = \iint \cos(\kappa)S_{ff} \left(\sqrt{\k^T \D^{-1} \k}\right) e^{ i \k^T \xib} d \k,
\end{equation}
and considering the change of variables
\begin{equation}
\k' = \U^T \k, \Rightarrow \k = \U \k', \Rightarrow d \k =| \mathrm{det}(\U)| d \k' = d \k',
\end{equation}
we find that 
\begin{equation}
r_{gg}(\xib) =  \iint \cos(\kappa) S_{ff} \left(\sqrt{\k'^T \Sigmab^{-1} \k'}\right) e^{i \k'^T \U^T \xib} d \k'.
\end{equation}
We would like to eliminate $\cos(\kappa)$. To this end, we need to rewrite $\cos(\kappa)$ as a function of the new variable. Any $2 \times 2$ orthogonal matrix may be expressed as either
\begin{equation}
\U_1 = \begin{bmatrix} \cos(\alpha) & -\sin(\alpha) \\ \sin(\alpha) & \cos(\alpha) \end{bmatrix} \quad \mbox{or} \quad
\U_2 = \begin{bmatrix} \cos(\alpha) & \sin(\alpha) \\ \sin(\alpha) & -\cos(\alpha) \end{bmatrix},
\end{equation}
where $\U_1$ represents a rotation and $\U_2$ a reflection. Thus, we find that
\begin{align}
\begin{bmatrix}  \cos(\kappa) \\ \sin(\kappa)\end{bmatrix} &=  \U_1 \begin{bmatrix}  \cos(\kappa') \\ \sin(\kappa')\end{bmatrix} =  \begin{bmatrix}  \cos(\alpha) \cos(\kappa') - \sin(\alpha) \sin(\kappa') \\ \sin(\alpha) \cos(\kappa') + \cos(\alpha) \sin(\kappa')\end{bmatrix},\\
\begin{bmatrix}  \cos(\kappa) \\ \sin(\kappa)\end{bmatrix} &=  \U_2 \begin{bmatrix}  \cos(\kappa') \\ \sin(\kappa')\end{bmatrix} =  \begin{bmatrix}  \cos(\alpha) \cos(\kappa') + \sin(\alpha) \sin(\kappa') \\ \sin(\alpha) \cos(\kappa') - \cos(\alpha) \sin(\kappa')\end{bmatrix}.
\end{align}
The covariance then becomes
\begin{align}
r_{gg}(\xib) &= \cos(\alpha) \iint \cos(\kappa) S_{ff} \left(\sqrt{\k^T \Sigmab^{-1} \k}\right) e^{i \k^T \U^T \xib} d \k \nonumber \\ & \phantom{=} \quad \mp \sin(\alpha) \iint \sin(\kappa) S_{ff} \left(\sqrt{\k^T \Sigmab^{-1} \k}\right) e^{i \k^T \U^T \xib} d \k,
\end{align}
where the negative sign corresponds to the rotation matrix and the positive sign to the reflection matrix. Now, let us apply the change of variables
\begin{equation}
\k' = \Sigmab^{-1/2} \k, \Rightarrow \k = \Sigmab^{1/2} \k', \Rightarrow d \k =| \mathrm{det}(\Sigmab^{1/2})| d \k' = d \k'
\end{equation}
to obtain
\begin{align}
r_{gg}(\xib) &= \cos(\alpha)\iint \cos(\kappa) S_{ff} \left(\|\k' \|\right)  e^{i \k'^T \Sigmab^{1/2} \U^T \xib} d \k' \nonumber \\ &\phantom{=} \quad \mp  \sin(\alpha) \iint \sin(\kappa) S_{ff} \left(\|\k' \|\right) e^{i \k'^T \Sigmab^{1/2} \U^T \xib} d \k',
\end{align}
where we notice that this expression also involves $\cos(\kappa)$ and $\sin(\kappa)$. These are given by 
\begin{align}
\cos(\kappa) &= \frac{k_1}{\sqrt{k_1^2 + k_2^2}} = \frac{\sigma_1^{1/2} k'_1}{\sqrt{\sigma_1 k_1^{'2} + \sigma_2 k_2^{'2}}}  = \frac{\cos(\kappa')}{\sqrt{ \cos^2(\kappa') + \sigma_1^{-2} \sin^2(\kappa')}}, \\
\sin(\kappa) &= \frac{k_2}{\sqrt{k_1^2 + k_2^2}} = \frac{\sigma_2^{1/2} k'_2}{\sqrt{\sigma_1 k_1^{'2} + \sigma_2 k_2^{'2}}}  = \frac{\sin(\kappa')}{\sqrt{\sigma_1^{2} \cos^2(\kappa') + \sin^2(\kappa')}}.
\end{align}
so the covariance finally becomes
\begin{equation}
r_{gg}(\xib)  = \iint \beta(\sigma_1,\alpha) S_{ff} \left(k\right) e^{i \k^T \tilde{\xib}} d \k,
\end{equation}
where
\begin{equation}
 \beta(\sigma_1,\alpha) =\frac{\cos(\alpha) \cos(\kappa) \sqrt{\sigma_1^{2} \cos^2(\kappa) + \sin^2(\kappa) } \mp \sin(\alpha) \sin(\kappa) \sqrt{ \cos^2(\kappa) + \sigma_1^{-2} \sin^2(\kappa)}}{ \sigma_1 \cos^2(\kappa) + \sigma_1^{-1} \sin^2(\kappa)}.
 \end{equation}
 Let us consider the Fourier series of $\beta(\lambda_1,\alpha)$, given by
 \begin{equation}
 \beta(\sigma_1,\alpha) = \sum_{l = -\infty}^{\infty} a_l(\sigma_1,\alpha) e^{i l \kappa},
 \end{equation}
 where the Fourier coefficients are 
 \begin{equation}
a_l(\sigma_1,\alpha) = \frac{1}{2 \pi} \int_{-\pi}^{\pi} \beta(\sigma_1,\alpha) e^{- i l \kappa} \, d \kappa.
 \end{equation}
Note that we cannot guarantee that the number of non-zero Fourier coefficients is finite for any given choice of $\lambda_1$ and $\alpha$. Plugging the Fourier series into the covariance, we get
\begin{equation}
r_{gg}(\xib)  = \sum_{l = -\infty}^{\infty} a_l(\sigma_1,\alpha) \int_{0}^{\infty} S_{ff} \left(k\right) k d k \int_{-\pi}^{\pi} e^{i l \kappa + i k \tilde{\xi} \cos(\kappa - \tilde{\xi})} d \kappa.
\end{equation}
Finally, considering the change of variable
\begin{equation}
\kappa' =  \frac{\pi}{2} + \tilde{\xi} - \kappa, \Rightarrow \kappa =  \frac{\pi}{2} + \tilde{\xi} - \kappa', \Rightarrow d \kappa = d \kappa',
\end{equation}
we may rewrite the covariance as\footnote{Note that we do not need to change the integral limits due to the periodicity of the involved functions.}
\begin{align}
r_{gg}(\xib)  &=  \sum_{l = -\infty}^{\infty} i^l e^{i l \tilde{\xi}} a_l(\sigma_1,\alpha) \int_{0}^{\infty} S_{ff} \left(k\right) k d k \int_{-\pi}^{\pi} e^{i k \tilde{\xi} \sin(\kappa') - i l \kappa' } d \kappa' \nonumber \\
&= \sum_{l = -\infty}^{\infty} i^l e^{i l \tilde{\xi}} a_l(\sigma_1,\alpha) \int_{0}^{\infty} S_{ff} \left(k\right) J_l ( k \tilde{\xi}) k d k 
\end{align}
which yields
\begin{equation}
r_{gg}(\xib)  =  \sum_{l = -\infty}^{\infty} i^l e^{i l \tilde{\xi}} a_l(\sigma_1,\alpha) \mathcal{H}_l ( S_{ff}, \tilde{\xi}).
\end{equation}

\section{Proof of Theorem \ref{th:c_eta_uni}}
\label{uni_proof}

If $f(\x)$ is unidirectional, we have seen that its monogenic signal admits the representation
\begin{equation}
m(\x)=f(\x)+{\eta}s(\x),
\end{equation}
where ${\eta}=\cos(\nu)i+\sin(\nu)j$. Define the second quaternion ${\eta}'=-\sin(\nu)i+\cos(\nu)j$, and augment the system with $k$.
Then
\begin{equation}
m^{(\eta')}(\x) = m^{(k)}(\x)  = m^\ast(\x).
\end{equation}
We know that any stationary random field satisfies
\begin{equation}
\cov\{m(\x),m^\ast(\x-\xib)\}=0,
\end{equation}
and so the sample from the monogenic unidirectional signal is $\mathbb{C}^{\eta}$-proper.

To prove the converse of the statement we assume
that we have the monogenic signal of a stationary random field, and that any finite sample from it is $\mathbb{C}^{\eta}$-proper.
Because it is a monogenic signal it has the form
\begin{equation}
m(\x)=f(\x)+i g(\x)+j h(\x),
\end{equation}
without a $k$-component. Signals that are $\mathbb{C}^{\eta}$-proper satisfy
\begin{align}
r_{m m^{(\eta')}} (\xib) = \cov\left\{{m}(\x),{m}^{(\eta')}(\x-\xib)\right\}={0},\quad
r_{m m^{(k)}} (\xib) = \cov\left\{{m}(\x),{m}^{(k)}(\x-\xib)\right\}={0}.
\end{align}
The second statement is true because we have a stationary monogenic signal, and therefore provides no additional information as ${m}^{(k)}(\x) = {m}^{\ast}(\x).$
We write the monogenic signal in the basis $\{\eta, \eta', k\}$ as
\begin{equation}
m(\x)=f(\x)+\eta\left(\cos(\nu)g(\x)+\sin(\nu) h(\x)\right)+\eta'
\left(-\sin(\nu)g(\x)+\cos(\nu) h(\x)\right).
\end{equation}
We therefore find that
\begin{multline}
r_{m m^{(\eta')}} (\xib) 
= r_{ff}(\xib)-\cos(2\nu)[r_{gg}(\xib)-r_{hh}(\xib)]-2\sin(2\nu)
r_{gh}(\xib)\\
+\eta'[\sin(\nu)2r_{fg}(\xib)-2\cos(\nu)r_{fh}(\xib)]
+k\{[r_{gg}(\xib)-r_{hh}(\xib)]\sin(2\nu)-2
r_{gh}(\xib)\cos(2\nu)\}.
\end{multline}
Since we are assuming $\mathbb{C}^{\eta}$-propriety, the complementary covariance satisfies $r_{m m^{(\eta')}} (\xib) = 0$, which yields
\begin{align}
\label{cond6}
\tan (\nu)&=\frac{r_{fh}(\xib)}{r_{fg}(\xib)},\\
\label{cond7}
r_{hh}(\xib)&=\sin^2(\nu)r_{ff}(\xib),\\
\label{cond8}
r_{gg}(\xib)&=\cos^2(\nu)r_{ff}(\xib),\\
\label{cond9}
r_{gh}(\xib)&=\sin(\nu)\cos(\nu)r_{ff}(\xib).
\end{align}
Consider a finite patch of the random field and assemble the samples of the random field and its Riesz transforms into the vectors $\f$, $\g = a_g \mathbf{f}+ b_g\mathbf{f}_g^{\perp}$ and $\h = a_h \mathbf{f}+ b_h\mathbf{f}_h^{\perp}$. The covariance matrices of the vectors $\mathbf{g}$ and $\mathbf{h}$ are
\begin{align}
\cov\left\{\mathbf{g},\mathbf{g}\right\}&=a_g^2\cov\left\{\mathbf{f},\mathbf{f}\right\}+
b_g^2\cov\left\{\mathbf{f}_g^{\perp},\mathbf{f}_g^{\perp}\right\}=\cos^2(\nu)\cov\left\{\mathbf{f},\mathbf{f}\right\},\\
\cov\left\{\mathbf{h},\mathbf{h}\right\}&=a_h^2\cov\left\{\mathbf{f},\mathbf{f}\right\}+
b_h^2\cov\left\{\mathbf{f}_h^{\perp},\mathbf{f}_h^{\perp}\right\}=\sin^2(\nu)\cov\left\{\mathbf{f},\mathbf{f}\right\},
\end{align}
which implies that the covariance matrices satisfy
\begin{align}
\cov\left\{\mathbf{f}_g^{\perp},\mathbf{f}_g^{\perp}\right\} = \cov\left\{\mathbf{f}_h^{\perp},\mathbf{f}_h^{\perp}\right\}= \cov\left\{\mathbf{f},\mathbf{f}\right\},
\end{align}
 and the scalars satisfy $a_g^2+b_g^2 = \cos^2(\nu)$ and $a_h^2+b_h^2 = \sin^2(\nu)$. Taking into account the covariance matrices between $\f$ and $\g$, and between $\f$ and $\h$, it is straightforward to show that $a_g = a_h = 0$, which implies that $b_g^2 = \cos^2(\nu)$, and $b_h^2 = \sin^2(\nu)$. Now, using the covariance matrix between $\g$ and $\h$, we have that $b_g = \cos(\nu)$, and $b_h = \sin(\nu)$. We shall call ${\mathbf{f}}_g^{\perp}=\mathbf{s}$. Then, the Riesz transforms are $\g = \cos(\nu)\mathbf{s}$, and $\h = \sin(\nu)\mathbf{s}$.
As we can write the sampled monogenic signal in this form for {\em arbitrary} sample length, we have $g(\x) = \cos(\nu) s(\x)$, and $h(\x) = \sin(\nu)s(\x)$, or in the frequency domain
\begin{align}
G(\bk)& = - i\cos(\kappa)F(\bk) = \cos(\nu) S(\bk), \\
H(\bk)& = - i\sin(\kappa)F(\bk) = \sin(\nu) S(\bk).
\end{align}
Finally, for this to hold for {\em all} $\bk$, we must have 
\begin{align}
F(\bk)&\propto \delta(\kappa-(\nu\pm m \pi))\\
S(\bk)&= -i \, {\mathrm{sgn}}\left(\mathbf{n}^T\bk\right)F(\bk).
\end{align}
This shows that the {\em only} monogenic stationary random field that is $\mathbb{C}^{\eta}$-proper is the unidirectional field.

\section{Proof of Lemma \ref{lem:det_plane_wave}}
\label{lem:det_plane_wave_proof}

Similar to $\hat{r}_{\widetilde{m} \widetilde{m}^{(\diamond)}}(\0)$, we define $\hat{r}_{ff}(\mathbf{0})$ as\footnote{There are similar definitions for $\hat{r}_{gg}(\x), \hat{r}_{hh}(\x)$, and $\hat{r}_{gh}(\x)$.}
\begin{equation}
\hat{r}_{ff}(\mathbf{0})=\frac{1}{N^2}\sum_{n,n'}f^2(\x_{n,n'}).
\end{equation}
The covariance of the periodic discrete Riesz transform $\tilde{g}(\x)$ is estimated as
\begin{equation}
\hat{r}_{\tilde{g} \tilde{g}}(\mathbf{0})=\frac{1}{N^2}\sum_{n,n'} \tilde{g}^2(\x_{n,n'})=\hat{r}_{gg}(\mathbf{0})+(\hat{r}_{\tilde{g} \tilde{g}}(\mathbf{0})-
\hat{r}_{gg}(\mathbf{0}))= \hat{r}_{gg}(\mathbf{0})+\Delta r_{\tilde{g} \tilde{g}},
\end{equation} 
with
\begin{equation}
\Delta r_{\tilde{g} \tilde{g}} = \frac{1}{N^2}\sum_{n,n'} \tilde{g}^2(\x_{n,n'}) - \frac{1}{N^2}\sum_{n,n'} g^2(\x_{n,n'}),
\end{equation}
and the covariances $\hat{r}_{\tilde{h} \tilde{h}}(\mathbf{0})$ and $\hat{r}_{\tilde{g} \tilde{h}}(\mathbf{0})$ are obtained analogously. Performing a Taylor series expansion of $\hat{\cal U}$ in terms of $\Delta r_{\tilde{g} \tilde{g}}, \Delta r_{\tilde{h} \tilde{h}}$ and $\Delta r_{\tilde{g} \tilde{h}}$, we find that 
\begin{equation}
\label{eq:U_taylor}
\hat{\cal U} = 1 - \frac{2 n_2^2 \Delta \hat{r}_{\tilde{g} \tilde{g}} + 2 n_1^2
\Delta \hat{r}_{\tilde{h} \tilde{h}} - 4 n_1 n_2 \Delta \hat{r}_{\tilde{g} \tilde{h}}}{r_{ff}(\0)} + {\cal O}\left(\max( \Delta r^{3}_{\tilde{(\cdot)} \tilde{(\cdot)}  })\right),
\end{equation}
where ${\cal O}\left(\max( \Delta r^{3}_{\tilde{(\cdot)} \tilde{(\cdot)}  })\right)$ stands for higher-order terms of $\Delta r_{\tilde{g} \tilde{g}}, \Delta r_{\tilde{h} \tilde{h}}$ and $\Delta r_{\tilde{g} \tilde{h}}$. We now need to quantify the deviation error $U_2 = 1 - \hat{\mathcal{U}}$. Let us start by writing the Fourier transforms of $g(\x)$ and $\tilde{g}(\x)$, which are given by
\begin{equation}
G(\bk) = R^{(1)}(\k) F(\k) =  G^{+}(\bk) + G^{-}(\bk)
\end{equation}
and
\begin{equation}
\tilde{G}(\bk) = \tilde{R}^{(1)}(\k) F(\k) = \tilde{G}^{+}(\bk) + \tilde{G}^{-}(\bk),
\end{equation}
respectively. Here, $G^{\pm}(\bk)$ and $\tilde{G}^{\pm}(\bk)$ are the infinite-length discrete-space and periodic discrete first Riesz transforms corresponding to the positive and negative frequencies, respectively. These are given by
\begin{equation}
G^{\pm}(\bk) = - i n_1\frac{A}{2}e^{i\phi}e^{i\frac{(N-1)[(k_0n_1 \mp k_1) + (k_0n_2 \mp k_2)]}{2}}
\frac{\sin(N/2(k_0n_1  \mp k_1))}{\sin(1/2(k_0
n_1 \mp k_1))}\frac{\sin(N/2(k_0n_2 \mp k_2))}{\sin(1/2(k_0n_2 \mp k_2))},
\end{equation}
and
\begin{equation}
\tilde{G}^{\pm}(\bk) = - i \frac{k_1}{\sqrt{k_1^2+k_2^2}} \frac{A}{2}e^{i\phi}e^{i\frac{(N-1)[(k_0n_1 \mp k_1) + (k_0n_2 \mp k_2)]}{2}}
\frac{\sin(N/2(k_0n_1  \mp k_1))}{\sin(1/2(k_0
n_1 \mp k_1))}\frac{\sin(N/2(k_0n_2 \mp k_2))}{\sin(1/2(k_0n_2 \mp k_2))}.
\end{equation}
We notice that the main difference between $G^{\pm}(\bk)$ and $\tilde{G}^{\pm}(\bk)$ is the first term, which is $n_1$ for the former and $k_1/\sqrt{k_1^2 + k_2^2}$ for the latter. The term $n_1$ comes from obtaining the Riesz transform of a complex exponential, which may be done in closed form. The term $k_1/\sqrt{k_1^2 + k_2^2}$ is precisely the filter that has to be applied to obtain a Riesz transform. Now, the Parseval-Rayleigh relationship allows us to express $\hat{r}_{gg}(\mathbf{0})$ in the frequency domain as 
\begin{multline}
\hat{r}_{gg}({\mathbf{0}})=\frac{1}{N^4} \sum_{k_1, k_2} \left|G^{+}(\bk) + G^{-}(\bk)\right|^2
=\frac{1}{N^4} \sum_{k_1, k_2} \left\{ \left|G^{+}(\bk)\right|^2 + \left|G^{-}(\bk)\right|^2 \right\} + \mathcal{O}(1/N^2)  \\
=\hat{r}^{+}_{gg}({\mathbf{0}}) + \hat{r}^{-}_{gg}({\mathbf{0}}) + \mathcal{O}(1/N^2),
\end{multline}
and similar expressions hold for $\hat{r}_{\tilde{g} \tilde{g}}({\mathbf{0}})$, $\hat{r}^{+}_{\tilde{g} \tilde{g}}({\mathbf{0}})$ and $\hat{r}^{-}_{\tilde{g} \tilde{g}}({\mathbf{0}})$. Therefore, the error term for the positive and negative frequencies becomes
\begin{equation}
\Delta \hat{r}_{\tilde{g} \tilde{g}}^{\pm}  = \hat{r}^{\pm}_{\tilde{g} \tilde{g}}({\mathbf{0}})-\hat{r}^{\pm}_{gg}({\mathbf{0}}) =
\frac{A^2}{4N^4}\sum_{\lambda_1, \lambda_2}\left[\frac{\lambda_1^2}{\lambda_1^2+\lambda_2^2}-n_1^2\right]
\frac{\sin^2(\pi N(\lambda_0 n_1 \mp \lambda_1))}{\sin^2(\pi (\lambda_0 n_1 \mp \lambda_1))}
\frac{\sin^2(\pi N(\lambda_0 n_2 \mp \lambda_2))}{\sin^2(\pi (\lambda_0 n_2 \mp \lambda_2))},
\end{equation}
where we have used $k_i = 2 \pi \lambda_i$, and the overall error is
\begin{equation}
\label{eq:error_g}
 \Delta \hat{r}_{\tilde{g} \tilde{g}} = \Delta \hat{r}_{\tilde{g} \tilde{g}}^{+} + \Delta \hat{r}_{\tilde{g} \tilde{g}}^{-} = \frac{A^2}{2 N^4}\sum_{\lambda_1, \lambda_2}\left[\frac{\lambda_1^2}{\lambda_1^2+\lambda_2^2}-n_1^2\right]
\frac{\sin^2(\pi N(\lambda_0 n_1 - \lambda_1))}{\sin^2(\pi (\lambda_0 n_1 - \lambda_1))}
\frac{\sin^2(\pi N(\lambda_0 n_2 - \lambda_2))}{\sin^2(\pi (\lambda_0 n_2 - \lambda_2))} + \mathcal{O}(1/N^2).
 \end{equation}
Applying the same procedure, we obtain similar expressions for the remaining error terms. The error in the covariance of $h(\x)$ is
\begin{equation}
\label{eq:error_h}
 \Delta \hat{r}_{\tilde{h} \tilde{h}} = \frac{A^2}{2 N^4}\sum_{k_1, k_2}\left[\frac{\lambda_2^2}{\lambda_1^2+\lambda_2^2}-n_2^2\right]
\frac{\sin^2(\pi N(\lambda_0 n_1 - \lambda_1))}{\sin^2(\pi (\lambda_0 n_1 - \lambda_1))}
\frac{\sin^2(\pi N(\lambda_0 n_2 - \lambda_2))}{\sin^2(\pi (\lambda_0 n_2 - \lambda_2))} + \mathcal{O}(1/N^2),
 \end{equation}
and for the cross-covariance it is
 \begin{equation}
 \label{eq:error_gh}
 \Delta \hat{r}_{\tilde{g} \tilde{h}} = \frac{A^2}{2 N^4}\sum_{\lambda_1, \lambda_2}\left[\frac{\lambda_1 \lambda_2}{\lambda_1^2+\lambda_2^2}- n_1 n_2\right]
\frac{\sin^2(\pi N(\lambda_0 n_1 - \lambda_1))}{\sin^2(\pi (\lambda_0 n_1 - \lambda_1))}
\frac{\sin^2(\pi N(\lambda_0 n_2 - \lambda_2))}{\sin^2(\pi (\lambda_0 n_2 - \lambda_2))} + \mathcal{O}(1/N^2).
 \end{equation}
Finally, the proof is concluded by inserting \eqref{eq:error_g}, \eqref{eq:error_h} and \eqref{eq:error_gh} into \eqref{eq:U_taylor} and taking into account that $r_{ff}(\0) = A^2/2$.

\section{Proof of Lemma \ref{lem:det_plane_wave_2}}
\label{lem:det_plane_wave_2_proof}

Let us consider ${\mathcal{B}}_{01}$, which is given by
\begin{multline}
{\cal B}_{01 } = \frac{2}{N^4} \sum_{|j_2|=N^{\beta}}^{N/2-1} \sum_{j_1=-L}^{L} \frac{\left\{(c_1/N+j_1/N) n_2-(c_2/N+j_2/N)n_1\right\}^2}{( \lambda_0 n_1+c_1/N+j_1/N)^2+( \lambda_0 n_2+c_2/N+j_2/N)^2}\\ \times\frac{\sin^2(\pi (c_1 + j_1 ))}{\sin^2(\pi (c_1/N+j_1/N))} \frac{\sin^2(\pi (c_2 + j_2))}{\sin^2(\pi (c_2/N+j_2/N))},
\end{multline}
where $L \ll N$ and $\beta$ is a real number close to one, but smaller. Taking into account $c_1, j_1 \ll N$ and $\sin(x) \approx x$ for small $x$, we find that
\begin{equation}
{\cal B}_{01 } = \frac{2}{N^2} \sum_{j_1=-L}^{L} \frac{\sin^2(\pi (c_1 + j_1))}{(\pi (c_1 + j_1))^2} \sum_{|j_2|=N^{\beta}}^{N/2-1}  \frac{(c_2/N+j_2/N)^2 n^2_1}{\lambda_0^2 n_1^2+(\lambda_0 n_2+c_2/N+j_2/N)^2} \frac{\sin^2(\pi (c_2 + j_2))}{\sin^2(\pi (c_2/N+j_2/N))}.
\end{equation}
It is easy to verify that for a relatively small value of $L$, the first sum in the above expression is almost one for every value of $c_1$. Moreover, $\sin^2(\pi (c_2 + j_2)) = \sin^2(\pi c_2)$, which yields
\begin{equation}
{\cal B}_{01 } \approx \frac{2}{N^2} \sum_{|j_2|=N^{\beta}}^{N/2-1}  \frac{(c_2/N+j_2/N)^2 n^2_1}{\lambda_0^2 n_1^2+(\lambda_0 n_2+c_2/N+j_2/N)^2} \frac{\sin^2(\pi c_2)}{\sin^2(\pi (c_2/N+j_2/N))}.
\end{equation}
The above sum may be considered as a Riemann integral, and therefore
\begin{equation}
{\cal B}_{01 } \approx \frac{2}{N} n^2_1 \sin^2(\pi c_2) \int_{|\lambda_2| = N^{\beta -1}}^{1/2}  \frac{(c_2/N+\lambda_2)^2 }{\lambda_0^2 n_1^2+(\lambda_0 n_2+c_2/N+\lambda_2)^2} \frac{1}{\sin^2(\pi (c_2/N+\lambda_2))} \, d \lambda_2.
\end{equation}
Applying the change of variable $\lambda = c_2/N + \lambda_2$, and observing that $N^{\beta -1} \approx 0$, we find that
\begin{equation}
{\cal B}_{01} \approx \frac{2}{N} n^2_1 \sin^2(\pi c_2) \left[ \int_{0}^{1/2}  \frac{1 }{\lambda_0^2 n_1^2+(\lambda_0 n_2 + \lambda)^2} \frac{\lambda^2}{\sin^2(\pi \lambda)} \, d \lambda + \int_{0}^{1/2}  \frac{1 }{\lambda_0^2 n_1^2+(\lambda_0 n_2 - \lambda)^2}\frac{\lambda^2}{\sin^2(\pi \lambda)} \, d \lambda \right].
\end{equation}
Finally, the proof follows from the symmetry of ${\cal B}_{10}$ with respect to ${\cal B}_{01}$.

\section{Proof of Theorem \ref{th:U2_plane_wave}}
\label{th:U2_plane_wave_proof}

In this appendix, we obtain
\begin{equation}
E[U_2] = E[\sin ^2(\pi c_2)  {\cal G}(\lambda_0,n_1,n_2) ] + E[\sin ^2(\pi c_1)  {\cal G}(\lambda_0,n_2,n_1) ],
\end{equation}
where the expectation is with respect to $\nu$, with $\n = [n_1, n_2]^T = [\cos \nu, \sin \nu]^T$. We first note that we may assume that $c_i$ is a random variable uniformly distributed on $[-1/2,1/2]$ and independent of $\nu$, which is approximately true for large values of $N$. Then, $E[U_2]$ becomes
\begin{align}
E[U_2] &= E[\sin ^2(\pi c_2)] E[{\cal G}(\lambda_0,n_1,n_2) ] + E[\sin ^2(\pi c_1)] E[{\cal G}(\lambda_0,n_2,n_1) ] \nonumber \\ &= \frac{1}{2}  E[{\cal G}(\lambda_0,n_1,n_2) + {\cal G}(\lambda_0,n_2,n_1) ].
\end{align}
We should therefore obtain a closed-form expression for the integrals ${\cal G}(\lambda_0,n_1,n_2)$ and ${\cal G}(\lambda_0,n_2,n_1)$. We shall start with the individuals integrals, which are solved in the following lemma.
\begin{lemma}
\label{lem:integral_g}
The integral ${\cal G}_{\pm}(\lambda_0,n_1,n_2)$ is approximately given by
\begin{multline}
{\cal G}_{\pm}(\lambda_0,n_1,n_2) \approx \frac{n_1^2}{6} + |n_1| \left( \frac{1}{\lambda_0 \pi^2} + \frac{\lambda_0}{3} \left( 2 n_2^2 -  1 \right) \right) \left[ \arctan \left(\frac{1}{2 \lambda_0 |n_1|} \pm \frac{n_2}{|n_1|}\right) \mp \arctan \left(\frac{n_2}{|n_1|}\right) \right] \\ \mp\frac{ \lambda_0 n_1^2 n_2}{3} \log\left(1 \pm \frac{n_2}{\lambda_0}  + \frac{1}{4\lambda_0^2}\right).
\end{multline}
\end{lemma}
\begin{IEEEproof}
As far as we know, there is no analytic solution for the integral \eqref{eq:integral_G}. Hence, we may consider the second-order Taylor series only of the second term of the integrand, and thus
\begin{equation}
{\cal G}_{\pm}(\lambda_0,n_1,n_2) \approx \frac{n_1^2}{\pi^2} \int_{0}^{1/2} \frac{1}{\lambda^2 \pm 2 \lambda_0 n_2 \lambda  + \lambda_0^2} \left(1 + \frac{\pi^2 \lambda^2}{3}\right) \,d \lambda,
\end{equation}
which may be split as 
\begin{equation}
{\cal G}_{\pm}(\lambda_0,n_1,n_2) \approx \frac{n_1^2}{\pi^2} \int_{0}^{1/2} \frac{1}{\lambda^2 \pm 2 \lambda_0 n_2 \lambda  + \lambda_0^2}  \,d \lambda + \frac{n_1^2}{3} \int_{0}^{1/2} \frac{\lambda^2}{\lambda^2 \pm 2 \lambda_0 n_2 \lambda  + \lambda_0^2}   \,d \lambda.
\end{equation}
By applying a partial fraction expansion to the integrand of the second integral, we may find the primitive functions, e.g., in \cite{abramowitz}, and by substituting the integration limits the proof follows.
\end{IEEEproof}
Once we have these integrals we may find their sum, which is presented next.
\begin{lemma}
\label{lem:integral_g_2}
${\cal G}(\lambda_0,n_1,n_2)$ is given by
\begin{multline}
{\cal G}(\lambda_0,n_1,n_2) = {\cal G}_{+}(\lambda_0,n_1,n_2) + {\cal G}_{-}(\lambda_0,n_1,n_2) = \frac{n_1^2}{3} + \frac{ \lambda_0 n_1^2 n_2}{3}  \log\left(\frac{ \displaystyle 1 - \frac{4 \lambda_0 n_2}{1 + 4\lambda_0^2}}{\displaystyle 1 + \frac{4 \lambda_0 n_2}{1 + 4\lambda_0^2}}\right)   + \\ |n_1| \left( \frac{1}{\lambda_0 \pi^2} + \frac{\lambda_0}{3} \left( 2 n_2^2 -  1 \right) \right) \left[ \arctan \left(\frac{4 \lambda_0 |n_1|}{4 \lambda_0^2 - 1}\right) + \pi \right].
\end{multline}
\end{lemma}
\begin{IEEEproof}
It is easy to show that
\begin{multline}
\label{eq:G_plus_minus_temp}
{\cal G}(\lambda_0,n_1,n_2) = {\cal G}_{+}(\lambda_0,n_1,n_2) + {\cal G}_{-}(\lambda_0,n_1,n_2) = \frac{n_1^2}{3} + \frac{ \lambda_0 n_1^2 n_2}{3} \left[ \log\left(1 - \frac{n_2}{\lambda_0}  + \frac{1}{4\lambda_0^2}\right) -  \log\left(1 + \frac{n_2}{\lambda_0}  + \frac{1}{4\lambda_0^2}\right)  \right] + \\ |n_1| \left( \frac{1}{\lambda_0 \pi^2} + \frac{\lambda_0}{3} \left( 2 n_2^2 -  1 \right) \right) \left[ \arctan \left(\frac{1}{2 \lambda_0 |n_1|} + \frac{n_2}{|n_1|}\right) +  \arctan \left(\frac{1}{2 \lambda_0 |n_1|} - \frac{n_2}{|n_1|}\right)  \right].
\end{multline}
Now, using
\begin{equation}
\label{eq:tangents}
\arctan \left(a\right) +  \arctan \left(b\right) = \arctan \left(\frac{a + b}{1 - a b}\right) \quad (\text{mod} \, \pi)
\end{equation}
the proof follows.
\end{IEEEproof}

Plugging in the values of ${\cal G}(\lambda_0,n_1,n_2)$ and ${\cal G}(\lambda_0,n_2,n_1)$ in $E[U_2]$, it becomes
\begin{multline}
E[U_2] =  \frac{1}{6} + \frac{ \lambda_0}{6} E \left[n_1^2 n_2 \log\left(\frac{ \displaystyle 1 - \frac{4 \lambda_0 n_2}{1 + 4\lambda_0^2}}{\displaystyle 1 + \frac{4 \lambda_0 n_2}{1 + 4\lambda_0^2}}\right) \right] + \frac{ \lambda_0}{6} E \left[ n_2^2 n_1  \log\left(\frac{ \displaystyle 1 - \frac{4 \lambda_0 n_1}{1 + 4\lambda_0^2}}{\displaystyle 1 + \frac{4 \lambda_0 n_1}{1 + 4\lambda_0^2}}\right)  \right]   + \\ \left( \frac{1}{2 \lambda_0 \pi^2} - \frac{\lambda_0}{6} \right)  E \left[|n_1| \arctan \left(\frac{4 \lambda_0 |n_1|}{4 \lambda_0^2 - 1}\right)\right] + \frac{\lambda_0}{3}  E \left[ |n_1| n_2^2 \arctan \left(\frac{4 \lambda_0 |n_1|}{4 \lambda_0^2 - 1}\right) \right]  + \\ \left( \frac{1}{2 \lambda_0 \pi^2} - \frac{\lambda_0}{6} \right) E \left[ |n_2| \arctan \left(\frac{4 \lambda_0 |n_2|}{4 \lambda_0^2 - 1}\right) \right] + \frac{\lambda_0}{3} E \left[ |n_2| n_1^2 \arctan \left(\frac{4 \lambda_0 |n_2|}{4 \lambda_0^2 - 1}\right)\right]  + \\ \left( \frac{1}{2 \lambda_0 \pi} - \frac{\lambda_0 \pi}{6} \right) E\left[|n_1|\right] + \frac{\lambda_0 \pi}{3} E\left[|n_1| n_2^2 \right] + \left( \frac{1}{2 \lambda_0 \pi} - \frac{\lambda_0 \pi}{6} \right) E\left[|n_2|\right] + \frac{\lambda_0 \pi}{3} E\left[|n_2| n_1^2 \right].
\end{multline}
To prove the theorem we need to find each of the above expectations w.r.t. $\nu$, which we will take as uniformly distributed between $0$ and $\pi$. 
\begin{lemma}
The value of the first expectation is
\begin{align}
 E \left[n_1^2 n_2 \log\left(\frac{ \displaystyle 1 - \frac{4 \lambda_0 n_2}{1 + 4\lambda_0^2}}{\displaystyle 1 + \frac{4 \lambda_0 n_2}{1 + 4\lambda_0^2}}\right) \right] = - \lambda_0 + \frac{4}{3} \lambda_0^3.
\end{align}
\end{lemma}
\begin{IEEEproof}
The expectation may be rewritten as
\begin{align}
 E \left[n_1^2 n_2 \log\left(\frac{ \displaystyle 1 - \frac{4 \lambda_0 n_2}{1 + 4\lambda_0^2}}{\displaystyle 1 + \frac{4 \lambda_0 n_2}{1 + 4\lambda_0^2}}\right) \right] &= \frac{1}{\pi} \int_{0}^{\pi} \cos^2 (\nu) \sin(\nu) \log\left(\frac{ \displaystyle 1 - \frac{4 \lambda_0 \sin(\nu)}{1 + 4\lambda_0^2}}{\displaystyle 1 + \frac{4 \lambda_0 \sin(\nu)}{1 + 4\lambda_0^2}}\right) \, d \nu \nonumber \\
 &= \underbrace{\frac{1}{4 \pi} \int_{0}^{\pi} \sin(\nu) \log\left(\frac{ \displaystyle 1 - \frac{4 \lambda_0 \sin(\nu)}{1 + 4\lambda_0^2}}{\displaystyle 1 + \frac{4 \lambda_0 \sin(\nu)}{1 + 4\lambda_0^2}}\right) \, d \nu}_{I_1} + \underbrace{\frac{1}{4 \pi} \int_{0}^{\pi} \sin(3\nu)  \log\left(\frac{ \displaystyle 1 - \frac{4 \lambda_0 \sin(\nu)}{1 + 4\lambda_0^2}}{\displaystyle 1 + \frac{4 \lambda_0 \sin(\nu)}{1 + 4\lambda_0^2}}\right) \, d \nu}_{I_2}.
\end{align}
For the sake of notational simplicity let us define $a = 4 \lambda_0/(1 + 4 \lambda_0^2)$. We may therefore write 
\begin{equation}
I_1 = \frac{1}{4 \pi} \int_{0}^{\pi} \sin(\nu) \log\left(\frac{ \displaystyle 1 - a \sin(\nu)}{\displaystyle 1 + a \sin(\nu)}\right) \, d \nu,
\end{equation}
and, using for instance \cite{abramowitz}, it is easy to show that
\begin{equation}
I_1 = \frac{1}{2 a} \left[ \sqrt{1 - a^2} - 1\right] = - \lambda_0.
\end{equation}
On the other hand, we have
\begin{equation}
I_2 = \frac{1}{4 \pi} \int_{0}^{\pi} \sin(3 \nu) \log\left(\frac{ \displaystyle 1 - a \sin(\nu)}{\displaystyle 1 + a \sin(\nu)}\right) \, d \nu,
\end{equation}
whose solution is
\begin{equation}
I_2 = \frac{1}{6 a^3} \left(4- 3 a^2 - \frac{a^4 - 5 a^2 + 4}{\sqrt{1 - a^2}} \right) = \frac{4}{3} \lambda_0^3,
\end{equation}
which concludes the proof.
\end{IEEEproof}
\begin{lemma}
The value of the second expectation is
\begin{align}
E \left[n_2^2 n_1 \log\left(\frac{ \displaystyle 1 - \frac{4 \lambda_0 n_1}{1 + 4\lambda_0^2}}{\displaystyle 1 + \frac{4 \lambda_0 n_1}{1 + 4\lambda_0^2}}\right) \right] = - \lambda_0 + \frac{4}{3} \lambda_0^3.
\end{align}
\end{lemma}
\begin{IEEEproof}
The proof follows along similar lines as the previous lemma.
\end{IEEEproof}
\begin{lemma}
The value of the third expectation is
\begin{equation}
E \left[|n_1| \arctan \left(\frac{4 \lambda_0 |n_1|}{4 \lambda_0^2 - 1}\right)\right] = - 2 \lambda_0.
\end{equation}
\end{lemma}
\begin{IEEEproof}
The expectation is given by
\begin{equation}
E \left[|n_1| \arctan \left(\frac{4 \lambda_0 |n_1|}{4 \lambda_0^2 - 1}\right)\right] = \frac{1}{\pi} \int_{0}^{\pi} |\cos(\nu)| \arctan \left(\frac{4 \lambda_0 |\cos(\nu)|}{4 \lambda_0^2 - 1}\right) \, d \nu.
\end{equation}
It is clear that the integrand is symmetric with respect to $\pi/2$, which yields
\begin{equation}
\frac{1}{\pi} \int_{0}^{\pi} |\cos(\nu)| \arctan \left(\frac{4 \lambda_0 |\cos(\nu)|}{4 \lambda_0^2 - 1}\right) \, d \nu = \frac{2}{\pi} \int_{0}^{\pi/2} \cos(\nu) \arctan \left(\frac{4 \lambda_0 \cos(\nu)}{4 \lambda_0^2 - 1}\right) \, d \nu,
\end{equation}
or equivalently
\begin{equation}
\frac{2}{\pi} \int_{0}^{\pi/2} \cos(\nu) \arctan \left(\frac{4 \lambda_0 \cos(\nu)}{4 \lambda_0^2 - 1}\right) \, d \nu = \frac{2}{\pi} \int_{0}^{\pi/2} \cos(\nu) \arctan \left(b \cos(\nu)\right) \, d \nu,
\end{equation}
where $b = 4 \lambda_0/(4 \lambda_0^2 - 1)$. Now, using \cite{abramowitz}, the integral can be written as
\begin{equation}
 \frac{2}{\pi} \int_{0}^{\pi/2} \cos(\nu) \arctan \left(b \cos(\nu)\right) \, d \nu =  \frac{1}{b} \left(\sqrt{b^2 + 1} - 1\right).
\end{equation}
Substituting the value of $b$, the proof follows.
\end{IEEEproof}
\begin{lemma}
The fourth expectation is given by
\begin{equation}
E \left[ |n_1| n_2^2 \arctan \left(\frac{4 \lambda_0 |n_1|}{4 \lambda_0^2 - 1}\right) \right] = - \frac{\lambda_0}{2} - \frac{2}{3} \lambda_0^3.
\end{equation}
\end{lemma}
\begin{IEEEproof}
The expectation is given by
\begin{equation}
E \left[ |n_1| n_2^2 \arctan \left(\frac{4 \lambda_0 |n_1|}{4 \lambda_0^2 - 1}\right) \right] = \frac{1}{\pi} \int_{0}^{\pi} |\cos(\nu)| \sin^2(\nu) \arctan \left(\frac{4 \lambda_0 |\cos(\nu)|}{4 \lambda_0^2 - 1}\right) \, d \nu,
\end{equation}
and taking into account the symmetry around $\pi/2$, it may be rewritten as
\begin{multline}
E \left[ |n_1| n_2^2 \arctan \left(\frac{4 \lambda_0 |n_1|}{4 \lambda_0^2 - 1}\right) \right] = \frac{2}{\pi} \int_{0}^{\pi/2} \cos(\nu) \sin^2(\nu) \arctan \left(\frac{4 \lambda_0 \cos(\nu)}{4 \lambda_0^2 - 1}\right) \, d \nu
= \\ \underbrace{\frac{1}{2 \pi} \int_{0}^{\pi/2} \cos(\nu) \arctan \left(\frac{4 \lambda_0 \cos(\nu)}{4 \lambda_0^2 - 1}\right) \, d \nu}_{I_3} -  \underbrace{\frac{1}{2 \pi} \int_{0}^{\pi/2} \cos(3 \nu) \arctan \left(\frac{4 \lambda_0 \cos(\nu)}{4 \lambda_0^2 - 1}\right) \, d \nu}_{I_4}.
\end{multline}
The solution to $I_3$ follows from the previous lemma, and the solution to $I_4$ is
\begin{align}
I_4 = \frac{1}{12 b^3} \left( 3 b^2+ 4 - \frac{{b}^{4} + 5{b}^{2} + 4}{\sqrt{{b}^{2}+1}}   \right),
\end{align}
which follows from \cite{abramowitz}.
\end{IEEEproof}

Similarly, we find that
\begin{align}
E \left[ |n_2| \arctan \left(\frac{4 \lambda_0 |n_2|}{4 \lambda_0^2 - 1}\right) \right] &= E \left[ |\sin (
\nu)| \arctan \left(\frac{4 \lambda_0 |\sin (\nu)|}{4 \lambda_0^2 - 1}\right) \right] = -2 \lambda_0,\\ 
E \left[ |n_2| n_1^2 \arctan \left(\frac{4 \lambda_0 |n_2|}{4 \lambda_0^2 - 1}\right)\right] &= E \left[ |\sin (\nu)| \cos^2 (\nu) \arctan \left(\frac{4 \lambda_0 |\sin (\nu)|}{4 \lambda_0^2 - 1}\right)\right] = -\frac{\lambda_0}{2} - \frac{2 \lambda_0^3}{3}.
\end{align}
The last expectations are given by
\begin{align}
E\left[|n_1|\right] &= E\left[|\cos(\nu)|\right] = \frac{2}{\pi}, \\
E\left[|n_1| n_2^2 \right] &= E\left[|\cos(\nu)| \sin^2(\nu) \right]  = \frac{2}{3 \pi},\\
E\left[|n_2|\right] &= E\left[|\sin(\nu)|\right] = \frac{2}{\pi}, \\ 
E\left[|n_2| n_1^2 \right] &= E\left[|\sin(\nu)| \cos^2(\nu) \right]  = \frac{2}{3 \pi}.
\end{align}
Finally, taking into account all individual integrals and after a lot of tedious algebra, the proof follows.

\section{Proof of Theorem \ref{th:U2_unidirectional}}
\label{th:U2_unidirectional_proof}

We shall start with the Fourier transform of $f(\x)$, which is
\begin{equation}
F(\lambdab) = \sum_{n, n'}f(\x_{n, n'})e^{-2\pi i\lambdab^T\x_{n, n'}},
\end{equation}
where $\k = 2 \pi \lambdab$. Using the spectral representation of the random field, it becomes\footnote{Contrary to the spectral representation in previous sections, the limits of the integral are $\pm 1/2$ rather than $\pm \infty$ since we are considering sampled (discrete) random fields.}
\begin{align}
F(\lambdab)  &=\int_{0}^{1/2}d{\cal Z}_f(\lambda)e^{i \pi (N-1) [(\lambda n_1- \lambda_1) + (\lambda n_2 - \lambda_2) ]} \frac{\sin(N\pi (\lambda n_1- \lambda_1))}{\sin(\pi(\lambda n_1- \lambda_1))}\frac{\sin(N\pi(\lambda n_2-\lambda_2))}{\sin(\pi(\lambda n_2 - \lambda_2))} \nonumber \\
&\phantom{=} + \int_{-1/2}^{0}d{\cal Z}_f(\lambda) e^{i \pi (N-1) [(\lambda n_1- \lambda_1) + (\lambda n_2 - \lambda_2) ]} \frac{\sin(N\pi(\lambda n_1- \lambda_1))}{\sin(\pi( \lambda n_1- \lambda_1))}\frac{\sin(N\pi(\lambda n_2- \lambda_2))}{\sin(\pi( \lambda n_2 - \lambda_2))} \nonumber\\
&=F^+(\lambdab)+F^-(\lambdab).
\end{align}
Moreover, it will be useful to know (to easily compare with $\hat{r}_{gg}(\0)$ and $\hat{r}_{hh}(\0)$)
\begin{multline}
E\left[\hat{r}_{ff}(\0)\right] ={r}_{ff}(\0)= \frac{1}{N^2} E\left[\sum_{n, n'} f^2(\x_{n, n'})\right] = \frac{1}{N^4} \sum_{\lambdab} \left\{ E\left[\left|F^{+}(\lambdab)\right|^2\right] + E\left[\left|F^{-}(\lambdab)\right|^2\right] \right. \\ \left. \phantom{E\left[\left|F^{+}(\lambdab)\right|^2\right]} + E\left[F^{+ \ast}(\lambdab) F^{-}(\lambdab)\right] + E\left[F^{+}(\lambdab) F^{- \ast}(\lambdab)\right]  \right\},
\end{multline}
where we have applied the Parseval-Rayleigh relationship, and the expectations, unless otherwise stated, assume a fixed direction. Taking into account that the spectral process is proper, i.e.,
\begin{align}
\cov\{d{\cal Z}_f(\lambda),d{\cal Z}_f(\lambda) \}&=S_{ff}(k)\;d\lambda\;d\lambda \delta(\lambda-\lambda),\\
\cov\{d{\cal Z}_f(\lambda),d{\cal Z}_f(-\lambda) \}&=\cov\{d{\cal Z}_f(\lambda),d{\cal Z}_f^\ast(\lambda) \}=0,
\end{align}
where we have used the fact that $f(\x)$ is real, we find that
\begin{equation}
E\left[\left|F^-(\lambdab)\right|^2\right] = E\left[\left|F^+(\lambdab)\right|^2\right],
\end{equation}
and
\begin{equation}
E \left[F^+(\lambdab)F^{-\ast}(\lambdab)\right]= E \left[F^{+\ast}(\lambdab)F^{-}(\lambdab)\right] = 0.
\end{equation}
This yields
\begin{equation}
\label{eq:temp_doubt}
E\left[\hat{r}_{ff}(\0)\right] = \frac{2}{N^4}\sum_{\lambdab}\int_{0}^{1/2}S(\lambda) d\lambda \frac{\sin^2(\pi N (\lambda n_1-\lambda_1))}
{\sin^2(\pi (\lambda n_1-\lambda_1))}\frac{\sin^2(\pi N (\lambda n_2-\lambda_2))}{\sin^2(\pi(\lambda n_2 - \lambda_2))}.
\end{equation}

Let us continue with the periodic discrete Riesz transform
\begin{equation}
\tilde{G}(\lambdab) = -i \frac{\lambda_1}{\sqrt{\lambda_1^2+ \lambda_2^2}}F(\lambdab) =  -i \frac{\lambda_1}{\sqrt{\lambda_1^2+ \lambda_2^2}} \left[F^+(\lambdab)+F^-(\lambdab) \right],
\end{equation}
and its square
\begin{equation}
\left|\tilde{G}(\lambdab)\right|^2 =  \frac{\lambda_1^2}{\lambda_1^2+\lambda_2^2}\left|F^+(\lambdab)+F^-(\lambdab) \right|^2,
\end{equation}
which allows us to write
\begin{multline}
\hat{r}_{\tilde{g} \tilde{g}}(\0) = \frac{1}{N^2} \sum_{n, n'} \tilde{g}^2(\x_{n, n'}) = \frac{1}{N^4}\sum_{\lambdab}\left|\tilde{G}(\lambdab)\right|^2 = \\ \frac{1}{N^4}\sum_{\lambdab}\frac{\lambda_1^2}{\lambda_1^2+\lambda_2^2}\left( \left|F^+(\lambdab)\right|^2+F^+(\lambdab)F^{-\ast}(\lambdab)
+F^{+\ast}(\lambdab)F^{-}(\lambdab)+\left|F^-(\lambdab)\right|^2\right).
\end{multline}
On the other hand, using $g(\x) = n_1 s(\x)$, we find that
\begin{equation}
\hat{r}_{g g}(\0) =  \frac{1}{N^2} \sum_{n, n'} g^2(\x_{n, n'}) = \frac{n_1^2}{N^4}\sum_{\lambdab} \left| S(\lambdab)\right|^2.
\end{equation}
Moreover, it is easy to show that $S(\lambdab) = - i F^{+}(\lambdab) + i F^{-}(\lambdab)$, which yields
\begin{equation}
\hat{r}_{gg}(\bm{0}) = \frac{n_1^2}{N^4}\left[\sum_{\lambdab} \left(\left|F_N^+(\lambdab)\right|^2+\left|F_N^-(\lambdab)\right|^2\right) \right]+{\cal
O}\left(N^{-2} \right).
\end{equation}
The expectation of the error between the covariance of the infinite-length discrete-space and the periodic discrete Riesz transforms becomes
\begin{multline}
E[\Delta r_{gg}] = E[\hat{r}_{\tilde{g} \tilde{g}}(\mathbf{0})] - E[\hat{r}_{gg}(\mathbf{0})] = \\
\frac{2}{N^4}\sum_{\lambdab}\int_{0}^{1/2}\left(\frac{\lambda_1^2}{\lambda_1^2+\lambda_2^2}-n_1^2\right) \frac{\sin^2(\pi N (\lambda n_1- \lambda_1))} {\sin^2(\pi (\lambda n_1- \lambda_1))}\frac{\sin^2(\pi N (\lambda n_2-\lambda_2))}{\sin^2(\pi(\lambda n_2 - \lambda_2))} S(\lambda) \, d\lambda.
\end{multline}
Similar terms may be obtained for $E[\Delta r_{hh}]$ and $E[\Delta r_{gh}]$. 

Using the Taylor series of $\hat{\mathcal{U}}$ given in \eqref{eq:U_taylor}, $U_2$ is given by
\begin{equation}
U_2 = \frac{2 n_2^2 \Delta \hat{r}_{\tilde{g} \tilde{g}} + 2 n_1^2
\Delta \hat{r}_{\tilde{h} \tilde{h}} - 4 n_1 n_2 \Delta \hat{r}_{\tilde{g} \tilde{h}}}{\hat{r}_{ff} (\0)} + {\cal O}\left(\max( \Delta r^{3}_{\tilde{(\cdot)} \tilde{(\cdot)}  })\right).
\end{equation}
Let us now write $\hat{r}_{ff} (\0) = E\left[\hat{r}_{ff}(\0)\right] + \delta \hat{r}_{ff}(\0)$ and perform a Taylor series expansion of $1 + \delta \hat{r}_{ff}(\0)/E\left[\hat{r}_{ff}(\0)\right]$ to find that
\begin{equation}
U_2 = \frac{2 n_2^2 \Delta \hat{r}_{\tilde{g} \tilde{g}} + 2 n_1^2
\Delta \hat{r}_{\tilde{h} \tilde{h}} - 4 n_1 n_2 \Delta \hat{r}_{\tilde{g} \tilde{h}}}{E\left[\hat{r}_{ff}(\0)\right]} \left(1 - \frac{\delta \hat{r}_{ff}(\0)}{E\left[\hat{r}_{ff}(\0)\right]} \right) + {\cal O}\left(\max( \Delta r^{3}_{\tilde{(\cdot)} \tilde{(\cdot)}  })\right).
\end{equation}
Now, taking the expectation of $U_2$ it follows that
\begin{equation}
E\left[U_2\right] = \frac{2 n_2^2 E[\Delta \hat{r}_{\tilde{g} \tilde{g}}] + 2 n_1^2
E[\Delta \hat{r}_{\tilde{h} \tilde{h}}] - 4 n_1 n_2 E[\Delta \hat{r}_{\tilde{g} \tilde{h}}]}{E\left[\hat{r}_{ff}(\0)\right]} + {\cal O}\left(\max( \Delta r^{3}_{\tilde{(\cdot)} \tilde{(\cdot)}  })\right).
\end{equation}
Plugging in the values for $\Delta \hat{r}_{\tilde{g} \tilde{g}}, \ldots$, it becomes
\begin{equation}
E\left[U_2\right] 
=\frac{2}{E\left[\hat{r}_{ff}(\0)\right]} \int_0^{1/2} \left[\sum_{\lambda_1=-N/2}^{N/2-1} \sum_{\lambda_2=-N/2}^{N/2-1} C(\lambda_1,\lambda_2)\right] S(\lambda) d \lambda + \mathcal{O}(1/N^2),
\end{equation}
where $C(\lambda_1,\lambda_2)$ is defined in Lemma \ref{lem:det_plane_wave}, with the exception that $\lambda_0$ is substituted by $\lambda$. Now, using Lemma \ref{lem:det_plane_wave_2}, it is possible to rewrite $E\left[U_2\right]$ as
\begin{equation}
E\left[U_2\right] 
=\frac{2}{N E\left[\hat{r}_{ff}(\0)\right]} \int_0^{1/2} \left[ \sin^{2}(\pi c_2) {\cal G}(\lambda,n_1,n_2) + \sin^{2}(\pi c_1) {\cal G}(\lambda,n_2,n_1)\right] S(\lambda) d\lambda + \mathcal{O}(1/N^2).
\end{equation}
Finally, taking the expectation also with respect to the direction, it becomes clear, using Lemmas \ref{lem:integral_g}, \ref{lem:integral_g_2} and Theorem \ref{th:U2_plane_wave}, that the expectation is
\begin{equation}
E\left[U_2\right] 
=\frac{1}{N} \frac{ \displaystyle \int_0^{1/2} \left[ \frac{4}{ \pi^2} \frac{1}{\lambda} - \frac{4}{9} \lambda +   \frac{1}{3} - \frac{4}{\pi^2} \right] S(\lambda) d\lambda}{\displaystyle \int_0^{1/2} S(\lambda) d\lambda} + \mathcal{O}(1/N^2),
\end{equation}
where we have used
\begin{equation}
E\left[\hat{r}_{ff}(\0)\right] = {r}_{ff}(\0) =  2 \int_0^{1/2} S(\lambda) d\lambda.
\end{equation}
This is equivalent to \eqref{eq:temp_doubt} because the method of moments estimator is unbiased for the variance.
        
\section*{Acknowledgments}

The work of S. Olhede was supported by the U. K. government, Engineering and Physical Sciences Research Council, under the High Dimensional Models for Multivariate Time Series Analysis leadership fellowship (EP/I005250/1). The work of P. Schreier was supported by the Alfried Krupp von Bohlen und Halbach Foundation, under its program ``Return of German scientists from abroad''. We would like to thank Kevin W. Lewis for the Venus topography data set, which can be found at \burl{http://geo.pds.nasa.gov/missions/magellan/shadr_topo_grav/index.htm}. Finally, we thank the anonymous referees for helpful suggestions on exposition and additional references, which improved the clarity of the paper.

\bibliographystyle{IEEEtran}
\bibliography{uni_bib}

\end{document}